\documentclass[letterpaper,11pt]{article} 
\usepackage[hyphens]{url}  
\usepackage{graphicx} 
\usepackage{times}
\usepackage{xspace,comment,calc}
\usepackage{amssymb,amsmath,amsthm,amsfonts,amstext}
\usepackage[shortlabels,inline]{enumitem}
\usepackage{multirow}
\usepackage{bbm}
\usepackage{booktabs}
\usepackage{placeins}
\usepackage{subcaption}
\usepackage{tikz}
\usepackage{thm-restate, thmtools}
\usetikzlibrary{positioning}
\usepackage{xcolor}
\usepackage{setspace}
\usepackage{multirow,makecell}

\usepackage[backref=page]{hyperref}

\usepackage{algorithm}
\usepackage{algorithmic}

\title{Constant-Factor Distortion Mechanisms for $k$-Committee Election%
\thanks{An extended abstract is to appear in the Proceedings of the 39th AAAI, 2025.
Work supported in part by C. Swamy's NSERC Discovery grant.}}
\author{
    Haripriya Pulyassary\thanks{{\tt hp297@cornell.edu}. School of ORIE, Cornell
      University, Ithaca, NY 14853, USA.} 
\and 
    Chaitanya Swamy\thanks{{\tt cswamy@uwaterloo.ca}. Dept. of Combinatorics \&
      Optimization, Univ. Waterloo, Waterloo, ON N2L 3G1, CANADA.}
}
\date{}

\makeatletter
\newcommand{\REPEATN}[1]{\ALC@it\algorithmicrepeat%
\ #1 \textbf{times}\begin{ALC@rpt}}%
\newcommand{\ENDREPEAT}{\end{ALC@rpt}\ALC@it\algorithmicend}%
\makeatother

\newfloat{mechanism}{tbhp}{lop}
\floatname{mechanism}{Mechanism}

\newenvironment{namedalg}[2][{}]
{\renewcommand{\thealgorithm}{#2}\begin{algorithm}[{#1}]}{\end{algorithm}}
\newenvironment{namedmech}[2][{}]
{\renewcommand{\themechanism}{#2}\begin{mechanism}[{#1}]}{\end{mechanism}}

\newtheorem{theorem}{Theorem}[section]
\newtheorem{lemma}[theorem]{Lemma}
{\theoremstyle{definition} 
\newtheorem{definition}[theorem]{Definition} 

\newtheorem{remark}[theorem]{Remark}
}

\newtheorem{claim}[theorem]{Claim}

\newtheorem{fact}[theorem]{Fact}

\DeclareMathOperator{\poly}{poly}
\DeclareMathOperator{\polylog}{polylog}
\DeclareMathOperator{\prev}{prev}

\newcommand{\alt}{\text{alt}_\sigma}

\newcommand{\topl}{\ensuremath{\text{Top}_\ell}}
\newcommand{\good}{\ensuremath{\mathsf{good}}}

\newcommand{\bad}{\ensuremath{\mathsf{bad}}}

\newcommand{\close}{\ensuremath{\mathsf{close}}}
\newcommand{\far}{\ensuremath{\mathsf{far}}}

\newcommand{\ring}{R}
\newcommand{\F}{A}

\newenvironment{proofof}[1]{\begin{proof}[Proof of {#1}]}{\end{proof}}

\DeclareMathOperator*{\Exp}{\mathbb{E}}

\newcommand{\E}[2][{}]{\ensuremath{{\textstyle{\Exp}_{#1}}\bigl[#2\bigr]}}

\newcommand{\R}{\ensuremath{\mathbb R}}
\newcommand{\Rp}{\R_{\geq 0}}

\newcommand{\A}{\ensuremath{\mathcal{A}}}

\newcommand{\C}{\ensuremath{\mathcal{C}}}

\newcommand{\T}{\ensuremath{\mathcal T}}

\newcommand{\M}{\ensuremath{\mathcal M}}

\newcommand{\Sc}{\ensuremath{\mathcal S}}

\newcommand{\OPT}{\ensuremath{\mathit{OPT}}}
\newcommand{\opt}{\OPT}

\newcommand{\sm}{\ensuremath{\setminus}}
\newcommand{\es}{\ensuremath{\varnothing}}
\newcommand{\assign}{\ensuremath{\leftarrow}}
\newcommand{\prefset}{\succeq}

\newcommand{\ceil}[1]{\ensuremath{\bigl\lceil#1\bigr\rceil}}

\newcommand{\gm}{\ensuremath{\gamma}}

\newcommand{\sse}{\subseteq}

\newcommand{\td}{\ensuremath{\tilde d}}

\newcommand{\tO}{\ensuremath{\widetilde O}}
\newcommand{\tdT}{\ensuremath{\widetilde T}}
\newcommand{\tdF}{\ensuremath{\widetilde\F}}
\newcommand{\tS}{\ensuremath{\widetilde S}}
\newcommand{\tG}{\ensuremath{\widetilde G}}

\newcommand{\bS}{\ensuremath{\overline S}}

\newcommand{\kp}{\ensuremath{\kappa}}
\newcommand{\al}{\ensuremath{\alpha}}
\newcommand{\tht}{\ensuremath{\theta}}

\newcommand{\dt}{\ensuremath{\delta}}

\newcommand{\sg}{\ensuremath{\sigma}}

\newcommand{\ve}{\ensuremath{\varepsilon}}

\newcommand{\cons}{\ensuremath{\triangleleft}}
\newcommand{\thresh}{\ensuremath{\zeta}}
\newcommand{\down}{\ensuremath{\downarrow}}
\DeclareMathOperator{\distortion}{distortion}

\newcommand{\topalt}{\ensuremath{\mathit{top}}}
\newcommand{\bottom}{\ensuremath{\mathit{bottom}}}
\newcommand{\core}{\ensuremath{\mathit{core}}}
\newcommand{\betaval}{\beta} 
\newcommand{\gbeta}{\ensuremath{3}}
\newcommand{\rbeta}{\ensuremath{4}}
\newcommand{\rbetap}{\ensuremath{2}}

\newcommand{\bthresh}{\ensuremath{\overline \thresh}}

\newcommand{\topDist}{\ensuremath{\nu}}
\newcommand{\lcore}{\ensuremath{\core_\ell(C^*)}}
\newcommand{\rl}{\ensuremath{r_\ell(C^*)}}
\newcommand{\error}{\ensuremath{\mathsf{Err}}}
\newcommand{\trued}{d}

\newcommand{\kmcf}{\ensuremath{k\text{-MCF}}}

\newcommand{\boruvka}{{\sc Boruvka}}
\newcommand{\boruvkagen}{{\sc Boruvka-Gen}}
\newcommand{\kcent}{{\sc $k$-center}}
\newcommand{\kmed}{{\sc $k$-median}}
\newcommand{\bbredn}{{\sc BB}}
\newcommand{\bbtopl}{{\ref{bb}-\topl}\xspace}
\newcommand{\meytopl}{{\sc Meyerson-Top$_\ell$}}
\newcommand{\meytoplgen}{{\sc Meyerson-Top$_\ell$-Gen}}
\newcommand{\bbmeyerson}{{\sc Meyerson-BB}}
\newcommand{\bbmeyersongen}{{\sc Meyerson-BB-Gen}}
\newcommand{\adsampltopl}{{\sc ADSample-Top$_\ell$}}
\newcommand{\adsampltoplgen}{{\sc ADSample-Top$_\ell$-Gen}}
\newcommand{\adsamplring}{{\sc AdSample-Ring}}
\newcommand{\adsamplagent}{{\sc SampleMech}}
\newcommand{\adsampltot}{{\sc SampleMech-Tot}}
\newcommand{\adsamplagentgen}{{\sc SampleMech-Gen}}

\begin{document}

\maketitle 

\begin{abstract}
In the $k$-committee election problem, we wish to aggregate the preferences of $n$ agents
over a set of alternatives and select a committee of $k$ alternatives that minimizes the
cost incurred by the agents. While we typically assume that agent preferences are captured
by a cardinal utility function, in many contexts we only have access to ordinal
information, namely the agents' rankings over the outcomes. As preference rankings are not
as expressive as cardinal utilities, a loss of efficiency is inevitable, and is quantified
by the notion of \emph{distortion}.

We study the problem of electing a $k$-committee that minimizes the sum of the 
$\ell$-largest costs incurred by the agents, when agents and candidates are embedded in a
metric space. This problem is called the $\ell$-centrum problem and captures both the
utilitarian and egalitarian objectives. When $k \geq 2$, it is not possible to compute a
bounded-distortion committee using purely ordinal information. We develop the first
algorithms (that we call mechanisms) for the $\ell$-centrum problem (when $k \geq 2$),
which achieve $O(1)$-distortion while eliciting only a very limited amount of cardinal
information via value queries. We obtain two types of query-complexity guarantees: $O(\log
k \log n)$ queries \emph{per agent},  and $O(k^2 \log^2 n)$ queries \emph{in total} (while
achieving $O(1)$-distortion in both cases). En route, we give a simple adaptive-sampling
algorithm for the $\ell$-centrum $k$-clustering problem.   
\end{abstract}

\section{Introduction} \label{intro}
In many applications, we wish to aggregate the preferences of agents in a given system and
select an outcome that maximizes social welfare (i.e. the total value gained by the
agents) or minimizes social cost (i.e. the total cost incurred by the agents). While we
typically assume that agent preferences are captured by a \emph{cardinal} utility function
that assigns a numerical value to each outcome, in many contexts we only have access to
\emph{ordinal} information, namely the agents' rankings over the outcomes. There are many
reasons why such situations may arise; perhaps the most prominent is that the agents
themselves may find it difficult to place numerical values on the possible outcomes. As
ordinal preference rankings are not as expressive as cardinal utilities, a loss of
efficiency in terms of the quality of the outcome computed is
inevitable. \cite{ProcacciaR06} introduced the notion of \emph{distortion} to quantify
the worst-case efficiency loss for a given social choice function.  

Much of the prior work has primarily considered the \emph{utilitarian} objective, which minimizes the sum of individual costs incurred by the agents. However, this {utilitarian} objective may not always be the appropriate choice. For instance,
in some settings (e.g. where fairness is an important consideration), we may instead wish to consider an \emph{egalitarian} objective and minimize the \emph{maximum} cost incurred by any agent.  Both objectives are special cases of the $\topl$ objective, which minimizes the sum of the $\ell$ largest costs incurred by agents: clearly, when $\ell=1$ and $\ell=n$, we recover the egalitarian and utilitarian objectives respectively. 

In this work, we study the $k$-committee election problem, wherein each agent has a
preference ordering over the set of candidates and we wish to elect a committee of $k$ candidates, so as to minimize the $\topl$-cost.  An instance of this problem 
$(\C, \F, \sigma)$ consists of a set of $n$ agents or voters $\C$, a set of $m$
alternatives or candidates $\F$,  and a preference profile (a tuple giving the preference ordering over $\F$, for each agent), $\sigma$.  In line with prior work,   we consider the {\em metric setting}, wherein agents and candidates correspond to  points in a metric space  specified by a distance function $d:\C\times\F\to \mathbb{R}_{\ge 0}$ satisfying the ``triangle'' inequality: for any $i,j\in\C$ and $a,b\in\F$, we have $d(i,a)\leq d(i,b)+d(j,b)+d(j,a)$. 
We slightly abuse notation and use $d$ to also denote the resulting metric. This assumption models many applications, including those where agents prefer alternatives that are ideologically similar to them: here
$d(i,a)$  can be interpreted as the \emph{ideological distance} between agent $i$ and candidate $a$. As the preference profile $\sigma$ arises from the distance function $d$, it must be that $d$ is \emph{consistent} with $\sigma$, denoted
$d\cons\,\sg$: that is, for any $i \in \C$ and $a, b \in \F$, if $i$ prefers $a$ over $b$, denoted $a \succeq_i b$, then $d(i, a) \le  d(i, b)$. 

A social choice function (SCF) $f$ for $k$-committee election maps a preference profile $\sigma$ to a set in $\F^k:=\{S\sse\F:|S|\leq k\}$. The cost incurred by an agent $i$, when a set $S$ of candidates is chosen, is given by $d(i,S):=\min_{a\in S}d(i,a)$, i.e., the distance to the closest alternative in $S$. Since $f$ does not know the cardinal information, one would not expect $f$ to output  the best solution for the given metric $d$, and the {\em distortion} \cite{ProcacciaR06} of $f$ quantifies the worst-case loss in solution quality that can
occur due to the fact that $f$ does not have cardinal information. More precisely, $\distortion(f)$ is the worst case ratio (over all instances) of the  cost of the solution output by $f$ over the optimal cost; formally 
\[ \text{distortion}(f) = \sup_{\sigma }\sup_{d \triangleleft\, \sigma}
\frac{\topl(d(\C, f(\sigma)))}{\min_{S\in\F^k}\topl(d(\C,S))},  \]  
where $d(\C,S)$ denotes the vector $\{d(i,S)\}_{i\in S}$ of agents' costs for $S\in\F^k$, and $\topl(d(\C,S))$ is the $\topl$-cost of this vector. Throughout, we use ``algorithm'' to refer to a procedure whose input includes complete
cardinal information, i.e., the metric $d$, and use the term ``mechanism'' when the
input includes only ordinal information given by the preference profile $\sg$.

\subsection{Our contributions} \label{contrib}
We initiate the study of low-distortion mechanisms for $k$-committee election under the
$\topl$ objective. 
The underlying problem can be equivalently viewed as a $k$-clustering problem (clustering agents/points around $k$ alternatives/centers), and we sometimes use the clustering-terminology, {\em $\ell$-centrum problem}, to refer to this problem. As noted earlier, $\ell$-centrum is a very versatile model, which
generalizes, and interpolates between, the classical and extensively-studied
{\em $k$-center} ($\ell=1$) and {\em $k$-median} ($\ell=|\C|=n$) problems. 
 Even for $k$-median, for any $k>1$, it is \emph{impossible} to obtain low-distortion mechanisms using just ordinal information; the distortion can be: (a) $\Omega(n)$ when $k=2$~\cite{anshelevichZhuKnownFacilities}, and (b) unbounded when $k>2$  (Theorem \ref{thm:kWin-unboundedDist}). In light of this, a natural question that arises is: {\em can one achieve meaningful distortion bounds (for $\ell$-centrum) by eliciting a small amount of cardinal information?}  

We answer this question {\em affirmatively}. 
One of the simplest ways of obtaining cardinal information, which was also considered
in some recent work on $k$-committee election~\cite{BurkhardtCFRSS24},
is via a {\em value query}, wherein we query $d(i,j)$ for an agent-alternative pair
$(i,j)$. 
Our chief contribution is to develop {\em constant-factor} distortion mechanisms for
the $\ell$-centrum problem using a very limited number of value queries. 

We consider two ways of measuring query-complexity: (1)
{\em per-agent} query complexity, which measures the maximum number of queries that any
single agent is asked; and (2) {\em total} (or {\em average}) query complexity, wherein we
bound the total number of queries elicited from the entire agent population.
We devise mechanisms that achieve $O(1)$ distortion and obtain strong bounds under both
query-complexity measures. 
We focus on the setting $\F=\C$, 
though some of our results apply more generally. 
We obtain per-agent query-complexity bounds of $O(\log k\log n)$ and 
$\tO\bigl(k\cdot\log(\min\{\ell,n/\ell\})\bigr)$ (Mechanisms~\ref{bb-meyerson}
and~\ref{adsample} respectively), where the $\tO(.)$ notation suppresses 
$O(\log\log k)$ factors.
Observe that the
latter bound is {\em independent of $n$}, for any fixed $\ell$ as also for large $\ell$ 
(including the case $\ell=n$); in particular, for any fixed $k$ and $\ell$, we only make a 
{\em constant} number of queries per agent. 
The algorithmic idea 
leading to the latter bound is fairly robust,
and we show that it can be implemented to also 
yield a total query-complexity bound of 
$\tO\bigl(k^2\log(\min\{\ell,n/\ell\})\log^2 n\bigr)$ (Mechanism~\ref{adsample-tot});
observe that this implies that the {\em average} query complexity goes down to $0$ as $n$
grows!  

\begin{table}[ht!]
\centering 
{\small
\begin{tabular}{llll}
\toprule
\multirow{2}{*}{\makecell[l]{{\bf Query-complexity} \\{\bf measure}}} & 
\multirow{2}{*}{\bf Bound obtained} & 
\multicolumn{2}{c}{\bf Setting} \\ \cline{3-4}
& & \boldmath $\F=\C$ & \boldmath $\F\neq\C$ \\ \hline
\multirow{2}{*}{\bf Per-agent queries} & 
$O(\log k\log n)$ & Mechanism~\ref{bb-meyerson} & Mechanism~\ref{bb-meyerson-gen} \\ \cline{2-4}
& 
$\tO\bigl(k\log(\min\{\ell,\frac{n}{\ell}\})\bigr)$ & 
Mechanism~\ref{adsample} & Mechanism~\ref{adsample-gen} \\ \midrule
\textbf{Total queries} & 
$O(k^2\log\ell\log^2 n)$ & Mechanism~\ref{adsample-tot} \\ \bottomrule 
\end{tabular}
}
\caption{Summary of our results. All mechanisms achieve $O(1)$ distortion. 
\label{res-table}}
\end{table}


Our mechanisms are randomized and achieve $O(1)$-distortion with constant success
probability.%
\footnote{We cannot detect if failure occurs, i.e., the distortion bound is not met, but
we can boost the success probability by repetition, since 
we can evaluate the cost of a solution using one query per agent and at most
$k\ell$ queries in total.} 
They can be modified to achieve $O(1)$-distortion \emph{in expectation} with
the same expected query-complexity bounds; 
this is discussed in Section~\ref{sec:inExp}.   

To our knowledge, 
these are the \emph{first} results establishing distortion upper bounds for
$\topl$ $k$-committee election for $k > 1$. Some of these results were obtained in a
preliminary form in~\cite{Pulyassary22}. Our results partially
answer an open question posed in~\cite{BurkhardtCFRSS24} of obtaining small distortion for
norm-based $k$-clustering objectives. While they 
consider a separate generalization of $k$-median, it is worth noting that for $k$-median,  
we obtain
{\em significantly improved guarantees} compared to~\cite{BurkhardtCFRSS24}: 
we obtain a true approximation, as opposed to bicriteria solutions, 
utilizing much fewer total number of queries,  
$O(k^2\log^3 n)$, 
as opposed to $O(k^4\log^5 n)$. 


\paragraph{Technical contributions and overview.}
We focus on the $\F=\C$ setting; 
in Section~\ref{extn}, we discuss extensions to the case $\F\neq\C$.
Table~\ref{res-table} summarizes our main results.
 
Our mechanisms consist of two chief ingredients.
First, we compute a coarse estimate that approximates the optimal $\ell$-centrum
value, $\OPT=\OPT_\ell$, within $\poly(n)$ factors (Section~\ref{optestim}).
We actually estimate the optimal $k$-center or $k$-median value, which suffices, since
all $\OPT_\ell$ values are within a factor of $n$ of each other: for $r\leq\ell$, we
have $\OPT_r\leq\OPT_\ell\leq\frac{\ell}{r}\cdot\OPT_r$.
We utilize different methods for this, which differ in terms of their query-complexity
bounds and the approximation quality of the estimate returned. We briefly discuss these
methods below, and state the guarantees obtained. 
%
\begin{enumerate}[label=(\alph*), topsep=0.4ex, itemsep=0.2ex, leftmargin=0ex, itemindent=*,
    widest=b] 
\item {\bf Boruvka mechanism.} 
In Section~\ref{app:boruvka}, we use Boruvka's algorithm for MSTs to find a 
{\em minimum-cost $k$-forest}, 
where a $k$-forest is a graph with $k$ components. This procedure, 
Mechanism~\ref{boruvka}, runs in $O(\log n)$
iterations and each iteration uses at most $1$ query per agent and merges every
component with its ``closest neighbor.''

\begin{restatable}{theorem}{boruvkamech} \label{boruvka-thm}
Mechanism~\ref{boruvka}
has $O(\log n)$ per-agent query complexity and
returns an estimate $B$ such that $\OPT\leq B\leq n^2\cdot\OPT$.
\end{restatable}

\item {\bf\boldmath $k$-center and $k$-median mechanisms.}
Here, we use certain approximation 
algorithms for $k$-center and $k$-median to obtain our estimate. These have the benefit 
that their query complexity is independent of $n$.
For small $\ell$, we use the well-known $2$-approximate $k$-center
algorithm~\cite{Gonzalez85}. As observed by~\cite{BurkhardtCFRSS24}, this can be
implemented using $O(k)$ per-agent queries and $O(k^2)$ total number of queries. 
For large $\ell$, we use $k$-means++~\cite{ArthurV07}, a randomized $O(\log k)$-approximation
algorithm for $k$-median that utilizes an elegant adaptive-sampling
approach. {\em Adaptive sampling} is actually a core-algorithmic idea underlying some of our
mechanisms (see below) that we adapt to directly handle the $\topl$-objective and
obtain good total-query complexity, but a vanilla implementation easily yields $O(k)$
per-agent complexity. 

\begin{restatable}{theorem}{kcentmedmech} \label{kcenter-thm} \label{kmedian-thm} 
In polynomial time, we can compute:
\begin{enumerate}[label=(\alph*), topsep=0.2ex, noitemsep, leftmargin=*]
\item an estimate $B_1$ such that 
$\OPT_\ell\leq B_1\leq 2\ell\cdot\OPT_\ell$ using $O(k)$ per-agent queries and $O(k^2)$
queries in total;  

\item an estimate $B_n$ such that 
$\OPT_\ell\leq B_n\leq 8(\ln k+2)\cdot\frac{n}{\ell}\cdot\OPT_\ell$ holds with probability
at least $1/2$ using $O(k)$ queries per agent.
\end{enumerate}  
\end{restatable}
\end{enumerate}

\smallskip\noindent
Second, and this is our chief technical contribution, we show how to leverage these
estimates of $\OPT$ in combination with algorithmic ideas developed in the cardinal
setting, to obtain mechanisms with $O(1)$ distortion and low query complexity. 
We develop two core algorithmic ideas.
\begin{enumerate}[label=\arabic*., topsep=0.4ex, itemsep=0.2ex, leftmargin=0ex, itemindent=*]
\item {\bf Black-box reduction (Section~\ref{blackbox})}. We present a simple, yet
quite versatile reduction that {\em transforms the ordinal problem to the cardinal
$\ell$-centrum problem} (i.e., where we know the metric)
using $\polylog(n)$ value queries while incurring an $O(1)$-factor loss in solution
quality. 
We can then utilize {\em any} $O(1)$-approximation algorithm for cardinal
$\ell$-centrum 
in a {\em black-box fashion} to obtain $O(1)$ distortion. 

The reduction proceeds by approximating the true metric $\trued$ by a sufficiently-close
metric $\td$; see Mechanism~\ref{bb} and Remark~\ref{bb-topl}. 
Given an estimate $B\in[\OPT,\al\cdot\OPT]$, we consider each agent $i$. 
Roughly speaking, 
we partition $\bigl[\frac{\ve B}{\al n},B\bigr]$ 
into intervals $(\thresh,(1+\ve)\thresh]$, where the $\thresh$ values increase by a
$(1+\ve)$-factor. 
For each value $\thresh$, we can use binary search on $i$'s preference relation to
find all points $a$ for which $\trued(i,a)\leq(1+\ve)\thresh$. This entire procedure
uses $O(\log^2 n)$ value queries from $i$. Now we simply find any metric $\td$ that
is consistent with this information, i.e., satisfies $\td(i,a)\in(\thresh,(1+\ve)\thresh]$
whenever $\trued(i,a)\in(\thresh,(1+\ve)\thresh]$ and $\td(i,a)\leq\frac{\ve B}{\al n}$
whenever $\trued(i,a)\leq\frac{\ve B}{\al n}$. It is not hard to argue that every solution  
has roughly the same cost under the $\td$ and $\trued$ metrics; 
hence, we can work with the metric $\td$!


To improve 
this to $O(\log k\log n)$ per-agent query complexity, we combine the above with a 
{\em sparsification} idea. 
We move to an instance with $O(k)$ distinct {\em weighted} points, losing
an $O(1)$-factor. Running the black-box reduction on this weighted instance now
only requires $O(\log k\log n)$ queries per agent, since for each $\thresh$ value, we only 
need to use binary search over $k$ points.
We obtain the sparse instance by computing a bicriteria solution for
$\ell$-centrum that opens $O(k)$ centers and achieves $O(1)$-approximation. We show that
``moving'' each 
point to its nearest center in this solution yields the desired sparse instance.
We adapt the algorithm of~\cite{Meyerson01} for facility location to the $\ell$-centrum
setting (Algorithm~\ref{meyerson-topl}), and show that by suitably using our estimate $B$,
we can obtain the desired bicriteria solution. 
Mechanism~\ref{bb-meyerson} describes the combined mechanism.

\item {\bf Adaptive sampling for \boldmath $\topl$-objective (Section~\ref{sec:adsample})}.
We obtain per-agent query complexity that is independent of $n$, and total
query-complexity bounds, by exploiting an elegant random-sampling approach called
adaptive sampling due to~\cite{AggarwalDK09} (see 
also~\cite{ArthurV07,OstrovskyRSS12}), which yields good bicriteria solutions for
$k$-median. 
In adaptive sampling, we pick a random point to add to the current center-set $S$ 
choosing point $i$ with probability proportional to $d(i,S)$, and we do this for
$O(k)$ iterations. Observe that one value query to each agent $i$ suffices to calculate  
$d(i,S)$, so this procedure 
uses $O(k)$ per-agent queries.  
The above approach does not directly work for $\ell$-centrum.  
But we show that, by capitalizing on an insight of~\cite{ChakrabartyS19} that enables us
to (roughly speaking) cast the $\topl$-objective as a $k$-median objective (see
Claim~\ref{claim:CS19proxy}), we can suitably modify the way the next center is sampled 
and adapt the approach to handle the $\ell$-centrum problem
(Algorithm~\ref{adsample-topl}).%
\footnote{In fact, we can extend adaptive sampling to handle the 
general {\em minimum-norm $k$-clustering} problem~\cite{ChakrabartyS19};
see~\cite{Pulyassary22}. The query complexity blows up prohibitively, but this is of 
interest in the cardinal setting.}
%
We need to run this modified adaptive sampling $\tO\bigl(\log(\min\{\ell,n/\ell\})\bigr)$
times, using information gleaned from the estimates of $\OPT$ returned by 
Theorem~\ref{kcenter-thm},
so this yields 
$\tO\bigl(k\log(\min\{\ell,n/\ell\})\bigr)$ per-agent query complexity
(Mechanism~\ref{adsample}). 

To obtain the total query bound stated in Table~\ref{res-table}, we execute adaptive
sampling slightly differently  (see Mechanism~\ref{adsample-tot}). Instead of querying agents outside of $S$, we query agents in $S$, and we compute the $d(\C,S)$ cost-vector approximately. As in the black-box reduction, we consider geometrically-increasing distance thresholds within a $\poly(n)$-bounded range, and for each threshold $\thresh$, we compute the ring of points $a\in\C$ for which $d(a,S)\in(\thresh,(1+\ve)\thresh]$. As before, this can be computed  via binary search on $j$'s preference relation for each $j\in S$, so this takes $O(|S|\log^2 n)$ queries in total. Now, we can treat all points within a ring as having roughly the same $d(a,S)$ value, so we can approximately implement adaptive sampling by choosing a ring with the appropriate probability and then a uniform point within the
ring. This yields the desired total query complexity.
\end{enumerate}

\subsection{Related work}
Distortion was first introduced and studied by \cite{ProcacciaR06}. Subsequent works \cite{anshelevichCopeland,munagalaImprovedMetricDistortion2019} studied the distortion of SCFs for \emph{single-winner elections} in the  metric setting and conjectured that there exists a deterministic SCF with distortion of at most 3. This conjecture was ultimately resolved by \cite{gkatzelisResolvingOptimalMetric2020}, who gave a deterministic 3-distortion social choice function. 

Furthermore, a series of work  \cite{anshelevichRandomizedSocialChoice2017,kempeCommunicationDistortionRandomness2020,gkatzelisResolvingOptimalMetric2020} culminating in the recent $2.74$-distortion (randomized) SCF by \cite{charikarBreakingMetricVoting2023} showed that randomized SCFs can achieve strictly better distortion bounds.   Here, a longstanding conjecture was that there exists a randomized SCF that achieves a distortion of 2; this conjecture was refuted independently by \cite{charikarMetricDistortionBounds2021} and \cite{pulyassaryRandomizedMetricDistortion2021}. This latter work also gave an LP to find an \emph{instance-wise optimal randomized SCF}; i.e., the LP computes, for a given instance, the randomized SCF with smallest distortion.
	
We study $k$-committee election for $k > 1$. Committee election problems have been
well-studied by the social choice community (see, for instance,
\cite{faliszewski2017multiwinner} and references therein). Low-distortion algorithms of variants of the committee-election problem have been studied in the
social-welfare-maximization setting \cite{borodin2022distortion} and
social-cost-minimization setting \cite{goelRelatingMetricDistortion2018,chenFavoritecandidateVotingEliminating2020,CaragiannisSV22}, however, these models are quite different from the one we consider.  
	 
 In stark contrast to single-winner elections, \cite{CaragiannisSV22} showed that
 the distortion of any $k$-committee election algorithm is unbounded in the
 cost-minimization setting (for $k > 2$). In light of this result, a natural
 question to ask is whether eliciting a small amount of additional cardinal
 information from the agents can yield better algorithms. This has been studied
 for single-winner elections
 \cite{amanatidisComparisonQuery,abramowitzAwarenessVoterPassion2019}, as well as other
 social choice problems, including matchings \cite{amanatidis2022don,anari2023distortion}.    
In the cost-minimization setting, when $\F = \C$, \cite{BurkhardtCFRSS24} present $O(1)$-distortion  mechanisms for $k$-committee election under the $(k,z)$-clustering objective, wherein one  seeks to minimize $\sum_{j \in \C} (d(j, 
 S))^z$. A special case of this problem, when $z \to \infty$, is the $k$-center
 problem, wherein one minimizes the maximum induced assignment cost. For the
 $k$-center problem, \cite{BurkhardtCFRSS24} give a 2-distortion algorithm
 requiring a total of $O(k^2)$ value queries.  They also give $O(1)$-distortion
 mechanisms for the general $(k, z)$-clustering problem; these are
 \emph{bicriteria} mechanisms, and consequently select a set of candidates of
 cardinality larger than $k$. Finally, as noted earlier, a preliminary version of these  results was obtained in \cite{Pulyassary22}.

\section{Preliminaries}\label{ch:prelim}
Recall that $\C$ is a set of $n$ agents or voters, and $\F$ is a set of $m$ alternatives or candidates. For $i\in\C$, and $a, b \in \F$, we say that $a \succeq_i b$ if agent $i$ prefers candidate $a$ to $b$.  Each agent $i$'s preference relation 
$\succeq_i$ induces a total order on $\F$. We denote the top choice of $i \in \C$ as
$\topalt(i, \succeq_i)$, or just $\topalt(i)$ when $\succeq_i$ is clear from the
context. Similarly, we denote the top choice of $i \in \C$ when restricted to 
$S \subseteq \F$ as \mbox{$\topalt_S(i,\succeq_i)$}, or just $\topalt_S(i)$ when
$\succeq_i$ is clear from the context. 
Analogously, we use $\bottom(i,\succeq_i)$ and $\bottom_S(i,\succeq_i)$ (abbreviated to
$\bottom(i)$, $\bottom_S(i)$ respectively) to denote the bottom choice of $i\in\C$ in
$\F$, and among $S\sse\F$ respectively.
Let $\prefset$ be the collection of all total orders on $\F$. A preference
profile is a tuple 
$\sigma=(\succeq_1, \ldots, \succeq_n)\in\prefset^n$.
As mentioned earlier, we consider the metric setting, where agents and candidates are
located in a metric space 
specified by a distance function $d:\C\times\F\mapsto\Rp$ that satisfies the triangle
inequality and is consistent with $\sg$, denoted $d\cons\,\sg$: for any $i\in\C$, and
$a,b\in\F$, if $a\succeq_i b$, then $d(i,a)\leq d(i,b)$. 

The solution-space of the $k$-committee election problem is the collection of subsets
of $\F$ of size at most $k$, denoted $\F^k$. Any $S\in\F^k$ induces a cost vector
$d(\C,S):=\{d(i,S)\}_{i\in\C}$, where $d(i,S):=\min_{a\in S}d(i,a)$ is the cost incurred
by $i$.
The {\em $\topl$-cost} of a vector $v\in\Rp^n$ is the sum of the $\ell$ largest entries of 
$v$: $\topl(v)=\sum_{i=1}^\ell v^\down_\ell$, where
$v^\down$ is the vector $v$ with entries sorted in non-increasing order. 
 
We consider $k$-committee election under the $\topl$ objective, i.e., the cost
of a solution $S\in\F^k$ is $\topl(d(\C,S))$, and we seek to find a minimum-cost solution;    
we often refer to this as the \emph{$\ell$-centrum problem}. 
The special cases where $\ell=1$ and $\ell=n$ correspond to the classical
{\em $k$-center} and {\em $k$-median} problems respectively.
We use $\OPT_\ell$ to denote the optimal $\ell$-centrum cost; we drop the subscript $\ell$
when it is clear from the context.
While $\ell$-centrum has been studied in the setting where the
metric $d$ is given, 
our focus is on devising mechanisms given the
ordinal information specified by $\sg$. In the absence of cardinal information, it is
inevitable that any social choice function $f:\prefset\mapsto\F^k$ must incur some loss in
solution quality. This loss is quantified using the notion of distortion.
	
\begin{definition}
Let $f:\prefset\to \F^k$ be a social choice function for $k$-committee election. 
The {\em distortion of $f$} is defined as
\[ \text{distortion}(f) := \sup_{\sigma }\sup_{d \triangleleft\, \sigma}
\frac{\topl(d(\C, f(\sigma)))}{\min_{S\in\F^k}\topl(d(\C, S))}.  \] 
\end{definition}

We seek mechanisms with low distortion, but as noted earlier, this is impossible given
{\em only} ordinal information, for any $k\geq 1$. 
Anshelevich et al.~\cite{anshelevichZhuKnownFacilities} gave an $\Omega(n)$ lower bound
for $k=2$, and
Theorem~\ref{thm:kWin-unboundedDist} below shows that, for $k\geq 3$, in fact 
{\em no bounded distortion is possible} given only ordinal information 
(strengthening the $\Omega(n)$ lower bound for $k=2$).
We note that this result also follows from~\cite{CaragiannisSV22} (who consider a
different problem). 

	
\begin{theorem}\label{thm:kWin-unboundedDist}
For $k$-median with $k\geq 3$, there exists an instance $(\C, \sigma)$ for which any social choice
function 
has unbounded distortion. 
\end{theorem}

\begin{proof}[Proof of Theorem~\ref{thm:kWin-unboundedDist}]
	Consider the following instance with four agents where the set of voters and candidates is $\mathcal{C} = \{w, x, y, z\}$. The preference rankings are 
	\begin{equation*}
		\begin{array}{llll}
			w: x \succeq y \succeq z &
			x: w \succeq y \succeq z  \\ 
			y: z \succeq x \succeq w &
			z: y \succeq x \succeq w 
		\end{array}
	\end{equation*}
	The following metrics $d_1$ and $d_2$ are consistent with this preference ranking: 
	
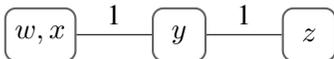
\begin{figure}[h!]
		\begin{subfigure}[b]{0.45\textwidth}
			\centering
			\begin{tikzpicture}[
				roundnode/.style={circle, draw=black!70,fill=white!70, thick, minimum size=3mm},
				squarednode/.style={rectangle, draw=black!60, fill=white!5,  thick, rounded corners, minimum size=7mm},
				]
				\node[squarednode]      (a)  {$w, x$};
				\node[squarednode]      [right=of a](c)  {$y$}; 
				\node[squarednode]      [right=of c](d)  {$z$};
				\draw (a)  -- (c) node [midway, above, fill=none]{1}
				(c) -- (d) node [midway, above, fill=none]{1};
			\end{tikzpicture}
			\caption{For any $i, j \in \mathcal{C}$, $d_1(i, j)$ is the shortest path distance in the above graph.}
		\end{subfigure}\hfill
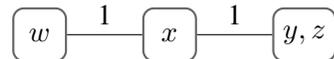
		\begin{subfigure}[b]{.45\textwidth}
			\centering
			
			\begin{tikzpicture}[
				roundnode/.style={circle, draw=black!70,fill=white!70, thick, minimum size=3mm},
				squarednode/.style={rectangle, draw=black!60, fill=white!5,  thick, rounded corners, minimum size=7mm},
				]
				\node[squarednode]      (a)  {$w$};
				\node[squarednode]      [right=of a](b)  {$x$};
				\node[squarednode]      [right=of b](c)  {$y, z$}; 
				\draw (a) -- (b) node [midway, above, fill=none]{1} (b) -- (c) node [midway, above, fill=none]{1};
			\end{tikzpicture}
			\caption{For any $i, j \in \mathcal{C}$, $d_2(i, j)$ is the shortest path distance in the above graph.}
		\end{subfigure}\vspace{10pt}	\caption{A $k$-winner selection instance with unbounded distortion}
	\end{figure}
	\FloatBarrier
	The optimal solution when considering $d_1$ is to choose $\{x, y, z\}$ as our committee -- this solution incurs a social cost of 0. Moreover, any other committee incurs a social cost of at least 1. On the other hand, the optimal solution under $d_2$ is to choose $\{w, x, y\}$ as our committee. This solution incurs a social cost of 0, and any other solution incurs a social cost of at least 1 (with respect to $d_2$). Since the (ordinal) information provided to us is insufficient to differentiate between $d_1$ and $d_2$, the distortion of any social choice $k$-correspondence is unbounded on this instance.
\end{proof}

Given these lower bounds, we focus on developing $O(1)$-distortion mechanisms
using a limited number of value queries. 
While different query models for eliciting cardinal information have enjoyed varying
levels of success for 
{\em social-welfare maximization}
problems~\cite{amanatidisComparisonQuery,maMenonLarson_thresholdQuery}, much less is known 
for the cost-minimization setting. One simple and very natural query is a 
{\em value query} (also used by~\cite{BurkhardtCFRSS24}), where we query
agent $i$ for the distance $d(i,a)$ between itself and alternative $a$. 
We consider mechanisms that utilize (a limited
number of) value queries, and extend the notion of distortion accordingly. 
	
\begin{definition} 
\label{defn:mechanism} 
A \emph{mechanism} $\mathcal{M}$ for $k$-committee election 
takes as input a preference profile $\sigma$,  
can adaptively make value queries, 
and outputs some solution $S\in\F^k$.  

The output of $\mathcal{M}$ can depend on $d$, but only via the answers of the value
queries made by it. We use $\mathcal{M}(\sigma|d)$ to denote the output of
$\mathcal{ M}$ on input $\sigma$ when the underlying metric is $d$. 

	
The {\em distortion of $\M$} is defined as:
\[ \text{distortion}(\mathcal{M}) := \sup_{\sigma }\sup_{d \triangleleft\, \sigma} 
\frac{ \topl(d(\C, \mathcal{M}(\sigma | d)))}{\min_{S \subseteq \F^k}\topl(d(\C, S))}  \]  
\end{definition}

\paragraph{Handling the $\topl$ objective.}
The $\topl$ objective can be difficult to work with due to its non-separable
nature: the contribution of an agent to the $\topl$-cost depends also on the other agents'
costs.  
We overcome this issue by working with the separable proxy function introduced
by~\cite{ChakrabartyS19}.  
For $z\in\R$, define $z^+:=\max\{z, 0\}$.  
	
\begin{claim} \cite{ChakrabartyS19} \label{claim:CS19proxy} \label{csproxy}		 
Let $v \in \Rp^n$ and $\rho \in \Rp$. Then,
(a) $\topl(v) \leq \ell \cdot \rho + \sum_{i=1}^n (v_i - \rho)^+$; and
(b) if $v^\downarrow_\ell \leq \rho \leq (1+\varepsilon)v^\downarrow_\ell$, we have
$\ell \cdot \rho + \sum_{i=1}^n (v_i - \rho)^+ \leq (1+\varepsilon) \cdot \topl(v)$. 	 
\end{claim}

By identifying a suitable value of $\rho$, 
we can work with the separable expression $\sum_i(v_i-\rho)^+$, (where $v$ is the cost
vector). This translates $\ell$-centrum into a $k$-median problem, albeit in
a non-metric setting, which allows us to exploit 
ideas used for $k$-median, for tackling the $\ell$-centrum problem.

\section{Computing estimates of $\OPT$}
\label{append-optestim} \label{optestim}

Our mechanisms crucially rely on having some coarse estimate of the optimal $\ell$-centrum
value, $\OPT=\OPT_\ell$. We present different methods for computing such an estimate,
differing in their query complexity and approximation quality. 
We consider the setting where $\F = \C$ here, and extend these to the setting $\F\neq\C$ 
in Section~\ref{optestim-gen}

\subsection{Boruvka mechanism} \label{app:boruvka}
One approach to compute such an estimate is to leverage the fact that $\opt$ is at least the cost of a minimum-cost $k$-forest, and is at most $n$ times the cost of a minimum-cost $k$-forest. 
\begin{claim}\label{claim:kMCF-kmedian}
	Let $\opt_{\kmcf}$ denote the cost of a minimum-cost $k$-forest, and let $\opt_n$ denote the cost of an optimal $k$-median solution. Then, \[ \opt_{\kmcf} \le \opt_n \le  n \cdot \opt_{\kmcf} \] 
\end{claim}
\begin{proof}
	Any $k$-median solution is a forest on $k$ components (where the edges are between each agent and its assigned cluster center), so $\opt_{\kmcf} \le \opt_n$. Let $F^*$ be a minimum cost $k$-forest. We can derive a $k$-median solution by choosing an arbitrary cluster center in each of the components induced by $F^*$, and assigning all clients in the cluster to this opened center. As we are preserving the components induced by $F^*$, due to the triangle inequality, the cost of this clustering is at most $n \cdot cost(F^*) = n \cdot \opt_{\kmcf}$. Thus, $\opt_n \le  \opt_{\kmcf}$. 
\end{proof}

So, if we knew $\opt_{\kmcf}$, the value of a minimum-cost $k$-forest, then $B = n \cdot \opt_{\kmcf}$ would satisfy $\opt \leq B \leq n \cdot \opt$. If $d(i, j)$ was known for all $i, j \in \mathcal{C}$, an optimal minimum-cost $k$-forest could be computed easily using Boruvka's algorithm. Boruvka's algorithm is a greedy minimum spanning tree (MST) algorithm, where at each stage, the cheapest edge incident to each (super)node is added and components are contracted into supernodes. The algorithm terminates when there is one supernode left. Given the MST, $T$, returned by Boruvka's algorithm (run with a fixed tie-breaking rule on the edges), we can remove the edges of $T$ in non-increasing order of cost, until we obtain a forest with exactly $k$ components; this is a minimum-cost $k$-forest.

Of course, we do not know $d(i, j)$ for all $i, j \in \mathcal{C}$. Querying the value of
$d(i, j)$ for all $i, j \in \mathcal{C}$ is computationally taxing on the agents, as this
would take $\Omega(n)$ queries per agent. However, in order to run Boruvka's algorithm, we
do not need to know the cost of \emph{all} edges; we only need to know the minimum cost
edge incident to each supernode. Hence, as we will show, only a few value queries are
needed to run Boruvka's algorithm. 
The precise algorithm is stated below, and leads to Theorem~\ref{boruvka-thm}, which we
restate below for convenience.
\begin{namedmech}[h!]{\boruvka}
	\caption{\hfill Minimum cost $k$-forest via Boruvka's algorithm}\label{alg:Boruvka}\label{boruvka}
	\begin{algorithmic}[1]
	\STATE Fix a tie-breaking rule on the edges (that will be used in all subsequent edge-cost comparisons). \\
	\STATE $F \leftarrow \varnothing$,  $V_1 \leftarrow \mathcal{C}$,  $E_1 \leftarrow \{ \{i, j \} : i, j \in \mathcal{C} \}$,  $t \leftarrow 1$ \\
	\WHILE{$|V_t| > 1$}{\label{line:whileStart}
	\FOR{$S \in V_t$}{
	\STATE For each $v \in S$, query the value of $\min_{e \in \delta(v) \cap \delta(S)}d(e)$\\
	\STATE Add $e = \arg\min_{e' \in \delta(S)} d(e')$ to $F$
	}\ENDFOR
	\STATE Contract the components of $G_t = (V_t, F \cap E_t)$ into supernodes to get the (multi)graph $G_{t+1} = (V_{t+1}, E_{t+1})$\\
	$t \leftarrow t + 1$
	}\ENDWHILE \label{line:whileEnd}
	\STATE Sort $F$ in non-increasing order of cost and remove edges in $F$ until exactly $k$ components are left \label{line:step14}\\
	\RETURN{$ n \cdot \sum_{e \in F} d(e)$} \label{line:step15}
	\end{algorithmic}
\end{namedmech}
\FloatBarrier

\boruvkamech*

\begin{proof} 
By Claim  \ref{claim:kMCF-kmedian}, $\opt \leq B \leq n^2 \opt$. It remains to show that
the number of queries elicited from each agent is at most $O(\log n)$.  

	Consider $S \in V_t$. For each $v \in S$, we know which edge attains $\min_{e \in \delta(v) \cap \delta(S)}d(e)$ (as we have the preference profile $\sigma $), so one value query is sufficient to compute the value of $\min_{e \in \delta(v) \cap \delta(S)}d(e)$. Given this, we can readily compute $e = \arg\min_{e' \in \delta(S)} d(e')$. Since each $v \in \mathcal{C}$ belongs to exactly one supernode of $V_t$, we incur the cost of one query per agent per iteration. 
	
	Since $|V_{t+1}| \le  \left\lceil\frac{|V_t|}{2}  \right\rceil$, the while-loop 
 terminates after $O(\log n)$ iterations; notice that the cost of every edge in $F$ is known, so no additional value queries are needed in the last two steps of the algorithm. 
 Thus, we make a total of $O(\log n)$ queries per agent. 
\end{proof}

\subsection{$k$-center and $k$-median mechanisms} \label{kcentkmed}

Mechanism~\ref{boruvka} has per-agent query complexity dependent on $n$. One can instead
use certain approximation algorithms for $k$-center and $k$-median to compute an estimate
of $\opt(d)$ with per-agent query complexity that is independent of $n$.  

For $k$-center, we use a well known deterministic $2$-approximation algorithm
of~\cite{Gonzalez85} that, at each step, opens a center at the client farthest from the
currently open centers. We observe that this can be implemented with low query complexity. 
Recall that $\topalt_S(j)$ and $\bottom_S(j)$ denote the top- and bottom-choice
alternatives of $j$ in a given set $S$, respectively.  

\begin{namedmech}[h!]{\kcent}
\caption{\hfill $2$-approximation for $k$-center~\cite{Gonzalez85}}
\label{alg:kcenter}\label{kcenter}\label{alg:kcenter:FC}\label{mech:kcenter}
				\begin{algorithmic}[1]
					\STATE $S_0 \leftarrow \varnothing$ \\
					\FOR{$t=1,\ldots, k$}
                        \STATE For each $i \in S_{t-1}$, query $d\bigl(i,\bottom_{\C_i}(i)\bigr)$, where $\C_i = \{j \in \C : \topalt_{S_{t-1}}(j) = i\}$
                        \STATE Choose $i^* \in \arg\max_{i \in S_{t-1}}d\bigl(i,\bottom_{\C_i}(i)\bigr)$, and set $s_t = \bottom_{\C_{i^*}}(i^*)$. 
                    \STATE Update $S_t \leftarrow S_{t-1} \cup \{s_t\}$. 
					\ENDFOR
					\RETURN $S_{k}, ~~ \max_{j \in \C} d(j, S_k)$\\
				\end{algorithmic}
			\end{namedmech}
   \FloatBarrier

When $\ell$ is large, we can obtain a better estimate using an algorithm for
$k$-median. 
The above algorithm does not perform well for $k$-median, but~\cite{ArthurV07} showed that
a randomized version of the algorithm, where we choose the next center to open
\emph{randomly} with probability proportional to the distance from the currently open
centers, returns a solution of 
expected cost of at most $O(\ln k) \cdot \opt_n$. (This was dubbed
adaptive sampling~\cite{AggarwalDK09}, and we discuss this in detail in
Section~\ref{sec:adsample}.) 

\begin{namedmech}[h!]{\kmed}
\caption{\hfill $O(\ln k$)-approximation for $k$-median~\cite{ArthurV07}}
\label{alg:AV}\label{kmedian}
				\begin{algorithmic}[1]
					\STATE $S_0 \leftarrow \varnothing$ \\
					\FOR{$t=1,\ldots, k$}
                    \STATE Query $d(j, \topalt_{S_{t-1}}(j))$ for $j \in \C \setminus S_{t-1}$. 
					\STATE Sample $s_t$ with probability proportional to $d(s_i, S_{t-1})$ \\ 
					\STATE Update $S_t \leftarrow S_{t-1} \cup \{s_i\}$. 
					\ENDFOR
					\RETURN $S_{k}, ~~ \sum_{j \in \C} d(j, S_k)$\\
				\end{algorithmic}
			\end{namedmech}

\noindent
Mechanisms~\ref{kcenter} and~\ref{kmedian} yield the bounds given in
Theorem~\ref{kcenter-thm}, which we restate below.

\kcentmedmech* 

\begin{proof} 
 Recall that $\OPT_\ell$ denotes the optimal $\ell$-centrum value. 
 For part (a), we take $B_1=\ell\cdot B'$, where $S_k, B'$ is the output returned by
 Mechanism~\ref{alg:kcenter}.  
 Gonzalez~\cite{Gonzalez85} proved that $B'\leq 2\cdot\opt_1$.
 Since $\opt_1 \leq \opt_\ell$
 and  $\opt_\ell \leq \ell \cdot \opt_1$, we obtain that $\opt_\ell \leq B_1 \leq 2 \ell
 \opt_\ell$.  As $\{j \in \C : \topalt_{S_{t-1}}(j) = i\}$ partitions $\C$, Mechanism
 \ref{alg:kcenter} elicits at most 1 query from each agent in each iteration, and
 consequently has a \emph{per-agent} query complexity of $k$. Furthermore, since
 $|S_{t-1}| \leq k$ at every step, the \emph{total} number of queries made is at most
 $k^2$. This was also observed by~\cite{BurkhardtCFRSS24}. 

  For part (b), \cite{ArthurV07} proved that the expected cost of the solution returned by
  Mechanism \ref{alg:AV} is at most $4(\ln(k) + 2) \cdot \opt_n$. Let  $S_k, B_n$ be the
  output returned by Mechanism \ref{alg:AV}. 
  Since $\opt_\ell \leq \opt_n \leq
  \frac{n}{\ell} \cdot \opt_\ell$ 
  we obtain that 
  $\opt_\ell \leq B_n \leq 4(\ln k + 2) \cdot \frac{n}{\ell} \cdot\opt_\ell$, in
  expectation. Moreover, Mechanism~\ref{alg:AV} has a per-agent query complexity of $k$.   
\end{proof}

\FloatBarrier

\section{Constant-factor distortion mechanisms when \boldmath $\F=\C$}\label{sec:A_eq_C}

\subsection{Black-box reduction: \boldmath $O(\log k\log n)$ per-agent queries} 
\label{sec:Blackbox} \label{blackbox}
When the 
metric is given as input, the $\ell$-centrum problem admits various
$O(1)$-factor approximation algorithms.
It would be ideal if we 
could somehow leverage this understanding of the cardinal problem. For instance, if we
could somehow reduce the ordinal setting to the cardinal setting, then we could utilize 
approximation algorithms developed in the cardinal setting to obtain low-distortion
mechanisms. 
A trivial such reduction utilizes queries $d(i,a)$ for every
$(i,a)\in\C\times\F$, but the question is: can we achieve this end using {\em
substantially} fewer queries.
We show that this is indeed possible. We give such a black-box reduction that makes only
$O(\log k\log n)$ per-agent queries, while losing only an $O(1)$-factor in the solution
quality; using any $O(1)$-approximation algorithm for cardinal $\ell$-centrum then yields
$O(1)$ distortion.  
	
We describe the idea for $k$-median, i.e., $\ell=n$, which extends with a very minor
change to the $\topl$ setting.
We consider a slightly more general setting, where each $i \in \mathcal{C}$ has an integer
weight $w_i\geq 0$ denoting the number of agents co-located with $i$; so $\sum_{i\in\C}w_i=n$ 
and the cost of a solution $S$ is $\sum_i w_i d(i, S)$. 
(This will enable us to handle sparsification seamlessly.)
%
As discussed earlier, we approximate the true underlying
metric $\trued$ by a close-enough metric $\td$; see Mechanism~\ref{bb}. 
Let $\OPT=\OPT_n(\trued)$ be the optimal value for
metric $\trued$, and $B\in[\OPT,\al\OPT]$ be an estimate. 
For each $i\in\C$ with $w_i>0$, we consider distance thresholds, roughly in the range
$\bigl[\frac{\ve B}{\al w_i n},\frac{B}{w_i}\bigr]$, and of the form
$\frac{B_{i,0}}{(1+\ve)^r}$ for integer $r\geq 0$, where $B_{i,0}$ is roughly
$\frac{B}{w_i}$. 
$\thresh$, we use binary search on $i$'s preference profile to find all points with
$j\in\C$ with $\trued(i,j)\leq\tau$. This takes $O(\log n)$ queries per threshold, and hence
$O(\log^2 n)$ queries to do this for all $\thresh$'s. Now, replacing
$\trued(i,j)\in(\thresh,(1+\ve)\thresh]$ by any value $\td(i,j)$ in this interval incurs only
a $(1+\ve)$-factor loss; similarly, if $\trued(i,j)\leq\frac{\ve B}{\al w_i n}$, then 
taking $\td(i,j)\leq\frac{\ve B}{\al w_i n}$ incurs an additive error of at most
$w_i\td(i,j)\leq\ve\OPT$. So for any $\td$ that is consistent with $d^*$ in this fashion,
the cost of any solution under $\td$ and $\trued$ is roughly the same. We can solve a linear
program (LP) to find such a consistent $\td$, and solve $k$-median with the metric $\td$.

For the $\topl$ objective, the only change to the above is that we replace $w_i$ by
$w'_i=\min\{w_i,\ell\}$; see Remark~\ref{bb-topl}.

	\begin{namedmech}[h!]{\bbredn}
		\caption{\hfill Blackbox reduction}\label{alg:blackboxRed} \label{bb} \label{alg:blackboxRedL}
		\textbf{Input:} 
{$(\C,\sg)$; integer weights $\{w_i\geq 0\}_{i \in \mathcal{C}}$ adding up to $n$;
estimate $B \in [\opt, \alpha \cdot \opt]$; 
$\rho$-approximation algorithm $\mathcal{A}$ for $k$-median} 
\\ 
		\vspace{-10pt}
		\begin{algorithmic}[1] 
			\FOR{$i \in \mathcal{C}$ with $w_i>0$}
			\STATE {Let $B_{i,0} = \rho(1+ 3 \varepsilon)\cdot \frac{B}{w_i}$, $q_i =\lceil \log_{1+\varepsilon }(\frac{\alpha w_i B_{i, 0} \cdot n}{\varepsilon B}) \rceil$} \\
                        \STATE {For each $r=0,\ldots,q_i$, use binary search to compute
                          $S_{i,r} = \{ j \in \mathcal{C} : \trued(i, j) \le B_{i,0}(1+\varepsilon )^{-r}\}$ 
                          in $O(\log n)$ queries}
			\ENDFOR
			\STATE Solve an LP to find a metric $\tilde d$ such that: 
			\begin{enumerate}
				\item[(1)] $\tilde d(i, j) \geq B_{i,0}$ for all $i\in\C$,
                                  $j \notin S_{i,0}$. 
				{\item[(2)] $(1+\varepsilon )^{-(r + 1)}B_{i,0} \leq
                                  \tilde d(i, j) \leq (1+\varepsilon )^{-r} B_{i,0}$ for
                                  all $i \in \mathcal{C}, r \in \{0, \ldots, q_i - 1\},j \in S_{i,r} \setminus S_{i, r + 1}$}
				\item[(3)]  $\tilde d(i, j) \leq \frac{\varepsilon B}{\alpha n \cdot w_i}$ for all $j \in S_{i,q_i}$
			\end{enumerate}
			\RETURN{$\mathcal{A}(\mathcal{C}, w, \tilde d)$}
		\end{algorithmic}
\begin{remark}[$\topl$-objective] \label{bb-topl}
The only change for the $\topl$ objective is that we replace $w_i$ by
$w'_i=\min\{w_i,\ell\}$ above (and of course $\A$ is now an algorithm for
$\ell$-centrum). We call the resulting mechanism, Mechanism~\bbtopl. 
\end{remark}
\vspace*{-1.25ex}
\end{namedmech}

\FloatBarrier

\begin{theorem}\label{thm:blackboxRed} \label{bb-thm} \label{bb-toplthm}
\label{thm:blackboxRedTopL}
Let $\trued$ be the true underlying metric, and let $\OPT_\ell(\trued)$ be the optimal
$\ell$-centrum cost for metric $\trued$. 
\begin{enumerate}[label=(\alph*), topsep=0.2ex, itemsep=0.2ex, leftmargin=*]
\item The center-set $F$ output by Mechanism~\ref{bb} satisfies
$\sum_{j \in \mathcal{C}}w_j\trued(j, F) \le  (\rho(1 + 2 \varepsilon)  +\varepsilon ) \opt_n(\trued)$ 

\item The output $F$ of Mechanism~\bbtopl satisfies
$\topl(\trued(\C,F|w)) \le  (\rho(1 + 2 \varepsilon ) + \varepsilon ) \opt_\ell(\trued)$, where 
$\trued(\C,F|w)$ is the vector in $\Rp^n$ obtained by creating $w_i$ coordinates of value
$\trued(i,F)$ for each $i\in\C$.
\end{enumerate}
%
Furthermore, 
these mechanisms can be implemented using 
$O(\log n \cdot \log(\alpha \rho \cdot  n)/\varepsilon)$ value queries per agent.
\end{theorem}

	Part (a) of Theorem~\ref{bb-thm} is a special case of part (b), so we focus on
        proving part (b), but the underlying ideas and 
	intuition do come from the $k$-median problem.
%
Let $\opt(\tilde d)$ denote the value of $\ell$-centrum for the metric $\td$.
	The following fact is immediate from the definition of $\tilde d$. Recall that
	$w'_i=\min\{w_i,\ell\}$. 
	
	\begin{fact}\label{fact:MetricClose}
	For any $i, j \in \mathcal{C}$, if $\trued(i, j) \le  B_{i, 0}$, then $\trued(i, j) - \kappa_i
	\le  \tilde d(i, j) \le  (1 + \varepsilon )\trued(i, j) + \kappa_i$, 
	where $\kappa_i = \varepsilon B/\alpha w'_i n$.
	\end{fact}
	
Given this, Claim~\ref{claim:metricCloseL} shows that if $T$ is a center-set
such that $\trued(i,T) \le  B_{i, 0}$ for all $i \in \mathcal{C}$, 
then the $\tilde d$-cost of $T$ is a good approximation of the $\trued$-cost of
$T$, and vice versa.  
Complementing this, Claim~\ref{claim:noBadEdgesL} shows that $\trued(i, T) \le B_{i, 0}$
for all $i \in \mathcal{C}$, 
for any $\rho$-approximate $\ell$-centrum solution $T$ for the metric $\tilde d$. 
Combining these claims yields the proof of Theorem~\ref{bb-thm}.

		\begin{claim}\label{claim:metricCloseL}
			Let $T \subseteq \mathcal{C}$ such that $\trued(i, T) \le  B_{i, 0}$ for all $i \in \mathcal{C}$. Then, 
			\begin{enumerate}[(a)]
				\item $\topl(\tilde d(\mathcal{C}, T| w)) \le  (1 + \varepsilon )\topl(\trued(\mathcal{C}, T| w)) + \varepsilon \opt(\trued)$
				\item $\topl(\trued(\mathcal{C}, T| w))  \le  \topl(\tilde d(\mathcal{C}, T| w))  + \varepsilon \opt(\trued)$
			\end{enumerate}
		\end{claim}
	\begin{proof} 
	Let $Q$ be a set of $\ell$ agents, where we take the weights, i.e., co-located clients
	into consideration; that is, more precisely, we take some $\gm_i\leq w_i$ points from each $i \in \C$, where
	$\sum_{i \in\C} \gm_i = \ell$. Note then that $\gm_i\leq w'_i$ for all $i\in\C$.
	Since $\trued(i, T) \le  B_{i, 0}$ for all $i \in \mathcal{C}$, by Fact \ref{fact:MetricClose}, 
	
	\begin{equation*}
	\sum_{i \in Q} \td(i, T) \le (1+\varepsilon )\sum_{i \in Q}\trued(i, T) + \sum_{i \in Q} w'_i \kappa_i
	\le  (1 + \varepsilon )\topl(\trued(\mathcal{C}, T | w)) + \sum_{i \in Q}
	w'_i \kappa_i  
	\end{equation*}
	
	\noindent Since $|Q| = \ell$, $\sum_{i \in Q} w'_i \kappa_i \le \ell\cdot\frac{\varepsilon \opt(\trued)}{n}$. 
	As this holds for any $\ell$-subset $Q$, we have 
	$\topl(\tilde d(\mathcal{C}, T| w)) \le  (1 + \varepsilon )\topl(\trued(\mathcal{C}, T|
	w)) + \varepsilon \opt(\trued)$. The proof of (b) is essentially the same. 
	\end{proof}
	
\begin{claim} \label{optrel}
We have	$\opt(\tilde d) \le (1 + 2\varepsilon)\cdot \opt(\trued)$. 
\end{claim}

\begin{proof}
Let $T\sse\C$ be 
an optimal $\ell$-centrum solution for $\trued$. We
have $d(i,T)\leq\OPT(\trued)\leq B\leq B_{i,0}$. The statement now follows from Claim
\ref{claim:metricCloseL} (a), since 
$\opt(\tilde d)\leq \topl(\tilde d(\mathcal{C}, T| w))$.
\end{proof}
 
	
\begin{claim}\label{claim:noBadEdgesL}
Let $T \subseteq \mathcal{C}$. If $\topl(\tilde d(\mathcal{C}, T | w)) \le \rho
\cdot \opt(\tilde d)$, then 
$\trued(i, T) \le  B_{i, 0}$ for all $i \in \mathcal{C}$.   
\end{claim}
	\begin{proof}
	Suppose, to arrive at a contradiction, that there exists $j \in \mathcal{C}$ such that
	$\trued(j, T) > B_{j, 0}$. Then, we also have $\td(j,T)\geq B_{j,0}$.
	Since $w'_{j} \le  \ell$, at least $w'_j$ agents who contribute to the $\topl$ objective
	incur a connection cost of $\td(j, T)$ or larger, so,  
\begin{equation*}
  \topl(\tilde d(\mathcal{C}, T | w)) |  \ge w'_{j}\cdot \tilde d(j, T) 
  \geq w'_{j} B_{j, 0} > \rho(1 + 2\varepsilon ) \opt(\trued) 
  \ge \rho \cdot \opt(\tilde d) 
\end{equation*}
			which is a contradiction. 
	\end{proof}
	
\begin{proof}[Proof of Theorem \ref{thm:blackboxRed}] 
Since $F$ is a $\rho$-approximate solution for the metric $\td$, by
Claim~\ref{claim:noBadEdgesL}, we have $\trued(i,F)\leq B_{i,0}$ for all $i\in\C$.
Now by Claim~\ref{claim:metricCloseL}, 
	we obtain 
	\begin{equation*}
	\begin{split}
	\topl(\trued(\mathcal{C},F|w)) & \le \topl(\td(\C,F|w))+\ve\OPT(\trued)
	\leq\rho\OPT(\td)+\ve\OPT(\trued) \\
	& \leq (\rho(1 + 2\varepsilon ) + \varepsilon )\opt(\trued) 
	\end{split}
	\end{equation*}
where we utilize Claim~\ref{optrel} for the final inequality.
This shows the stated performance guarantee.
		
	\medskip
	\noindent \emph{Query Complexity}: Mechanism \ref{alg:blackboxRed} uses queries to
	determine $S_{i, r}$ for all $i \in \mathcal{C}, r = 0, \ldots, q_i$. As we have the
	preference ranking for each agent, we have a list of agents sorted in non-decreasing order
	of their distance from $i$. Hence, to compute $S_{i, r}$, we can use binary search to
	determine maximal $p_1, p_2$ such that $p_1 < p_2$ and $\trued(i, \alt(p_1)) \le   B_{i, 0}
	(1+\varepsilon )^{-r} \le  \trued(i, \alt(p_2)) $. Then, $S_{i, r} = \{j \in \mathcal{C} :
	\alt(p_1) \succeq_i j \succeq_i \alt(p_2)\}$.  The total number of value queries required
	to compute $S_{i, r}$ in this manner is $O(\log n)$, and hence the total number of value
	queries (per agent) that is needed to determine each of $S_{i, 0}, \ldots, S_{i, q_i}$ for
	a fixed agent $i$  is $O(q_i \cdot \log n) = O(\log(n) \cdot \log(\alpha \rho \cdot  n) /
	\varepsilon)$.
	\end{proof}
	
	We obtain the estimate $B$ required by these mechanisms using Mechanism~\ref{boruvka},
	which yields $\al=n^2$ (see Theorem~\ref{boruvka-thm}). So the combined mechanism, with an
	$O(1)$-approximation algorithm for $\ell$-centrum, has $O(1)$ distortion and $O(\log^2 n)$
	per-agent query complexity.

The following slightly more-general guarantee for Mechanism~\bbtopl will be useful later (particularly when analyzing Mechanism~\ref{bb-meyerson}).
The proof is essentially the same as that of Theorem~\ref{thm:blackboxRedTopL}, and is omitted.  

\begin{theorem}\label{thm:blackboxRedL2}
Suppose the quantity $B$ in Mechanism~\bbtopl satisfies $B\in[U,\al U]$, for some
$U\geq\OPT(\trued)$, where $\trued$ is the true underlying metric.
The center-set $F$ output by the mechanism satisfies
	\[ 
		\topl(\trued(j, F|w))\le  (\rho(1 + 2 \varepsilon ) + \varepsilon )U
	\]
and the mechanism has $ O(\log(n) \cdot \log(\alpha \rho \cdot  n) / \varepsilon)$
per-agent query-complexity.
\end{theorem}

The difference between the statements of Theorem~\ref{bb-toplthm} and Theorem~\ref{thm:blackboxRedL2} 
is that in the latter we do not assume that $B$ estimates $\OPT(\trued)$ within any factor;
indeed, $B$ and $U$ could be quite large compared to $\OPT(\trued)$, and correspondingly we
only compare our solution quality to $U$, not $\OPT(\trued)$.

\paragraph{Improved \boldmath $O(\log k\log n)$ per-agent query complexity via
  sparsification.} 
One of the $\log n$-factors in the $O(\log^2 n)$ per-agent query complexity that we obtain
via Theorem~\ref{bb-toplthm}, arises
because we need to do binary search over $n$ agents to compute $S_{i,r}$.
To improve this,
we {\em sparsify} our instance before applying the black-box reduction. 
We do so by computing a $(\beta,\gm)$-bicriteria solution for $\ell$-centrum (using few
queries per agent), where we open at most $\beta k$ centers and incur cost at most $\gm$
times the optimum, and ``moving'' each agent to its nearest center in the bicriteria
solution. Suppose we have $\beta,\gm=O(1)$.
Then, we obtain a weighted instance with $O(k)$ points, and we argue that the move to the
weighted instance incurs only an $O(1)$-factor loss. Combining this with the earlier
black-box reduction now yields $O(\log k\log n)$ per-agent query complexity.

We compute an $\bigl(O(1),O(1)\bigr)$-bicriteria solution by extending the algorithm
of~\cite{Meyerson01} for facility location to the $\ell$-centrum setting.
In {\em facility location} (FL), 
any number of facilities may be opened, but every facility has an opening cost $f$, and we
seek to minimize the sum of the assignment costs and the facility-opening costs.
{Meyerson's algorithm for FL considers agents appearing
  online; when the $i$th client arrives at location $x_i$, it opens a facility at $x_i$
  with probability $\delta_i/f$, where $\delta_i$ is the distance from $x_i$ to the
  closest currently open facility. Meyerson proves, among other things, that when agents
  appear in a uniform random sequence, for every cluster $O^*$ in an optimal
  solution with corresponding center $c^*\in\C$, the random solution $S$ returned
  satisfies 
$\mathbb{E}\bigl[|S \cap O^*|f + \sum_{j \in O^*} d(j, S) \bigr]\le  5f + 8 \sum_{j \in O^*} d(j,o)$. 
Furthermore, this algorithm yields an $\bigl(O(1),O(1)\bigr)$-bicriteria solution for
$k$-median if {$f = B/k$}, where {$B$} is a $\Theta(1)$-estimate of optimal $k$-median
cost.}  

{We adapt Meyerson's algorithm and analysis to the $\topl$-setting, using the
  {separable} proxy function $\sum_{j \in \C} (d(j, S) - t)^+$ suggested by
  Claim~\ref{csproxy}; see Algorithm~\ref{meyerson-topl}. 
Viewing $(d(j, S) - t)^+$ as the proxy-cost of
agent $j$, 
$k$-clustering to minimize the proxy function gives
another type of $k$-median problem.
However, the proxy costs do not satisfy the triangle inequality, and to
circumvent complications arising from this, 
we actually work with the quantity $\delta_j := (d(j, S) - 3t)^+$, 
and as in Meyerson's algorithm, open a center at $j$ with probability $\delta_j/f$. 

		\begin{namedalg}[h!]{\meytopl}
			\caption{\hfill Meyerson's algorithm for FL adapted to $\ell$-centrum}
                        \label{alg:Meyerson-topL} \label{meyerson-topl}
			\textbf{Input}: {Sequence of agents $x_1,\ldots,x_n$, estimate
                          $B\geq\OPT$}\\ 
			\vspace{-10pt}
			\begin{algorithmic}[1]
				\STATE  $S \leftarrow \{x_1\}$, {$f = \frac{B}{k}$}\\
				\FOR{$i=2,\ldots, n$}
				\STATE $\delta_i = \left(d(x_i, S) - 3 \cdot {\frac{B}{\ell}}\right)^+$\\
				\STATE Add $x_i $ to $S$ with probability $\min(1, \delta_i/f)$\\
				\ENDFOR
				\RETURN {$S$}
			\end{algorithmic}

\begin{remark}
We have assumed above that the metric $d$ is given. But if we are only given a
preference profile, we can compute $\dt_i$ using one value query to $i$, so the
resulting mechanism has {\em unit} per-agent query complexity.
\end{remark}
\vspace*{-1.25ex}
\end{namedalg}

\FloatBarrier

\begin{theorem}\label{thm:Meyerson-lCentrum} \label{meyerson-topl-thm}
If the order of agents is random, the expected number of facilities opened by Algorithm
\ref{alg:Meyerson-topL} is at most $26 k$,  and the expected cost is at most $15B +
14\opt$. 
\end{theorem}

We defer the proof of Theorem \ref{thm:Meyerson-lCentrum} to
Section~\ref{mey-toplproof}, and show here how to leverage this to obtain $O(1)$
distortion using $O(\log k\log n)$ per-agent queries. 
Given an $O(1)$-estimate of $\OPT$,
Algorithm~\ref{meyerson-topl} yields an $\bigl(O(1),O(1)\bigr)$-bicriteria solution. We do
not have such an estimate, but we 
do have $B'\in[\OPT,n^2\cdot\OPT]$, and if we try all powers of $2$ in the range
$[B'/n^2,B']$, we will find some value in the range $[\OPT,2\cdot\OPT]$.
Also, Algorithm~\ref{meyerson-topl} may fail with some probability, so we boost its
success probability by repetition. Assuming we find the desired bicriteria solution, we
move to the weighted instance described earlier, and run the black-box reduction on this
weighted instance. Mechanism~\ref{bb-meyerson} puts together all of these ingredients, and
Theorem~\ref{bb-meyerson-thm} states its performance guarantee.

		\begin{namedmech}[h!]{\bbmeyerson}
                  \caption{\hfill $O(\log k \log n)$- per-agent query complexity}
			\label{alg:mech1L}  \label{bb-meyerson}
			\textbf{Input}: {Preference profile $\sigma$, $\rho$-approximation
                          algorithm $\mathcal{A}$ for $\ell$-centrum, where $\rho=O(1)$}\\
			\vspace{-10pt}
			\begin{algorithmic}[1]
   				\STATE $\mathcal{S} \leftarrow \{S_0\} $ where $S_0$ is an arbitrary set of $k$ centers\\
				\STATE $B'$: Output of Mechanism~\ref{boruvka}   \label{line:mech1L-inexp1-start} 
				\STATE $x_1, \ldots, x_n$: Randomly shuffled sequence of agents\\ 
				\FOR{$i=0,\ldots, \lceil \log_2 n^2 \rceil$} 
				\STATE $B_i \leftarrow 2^{i} \cdot B'/n^2$, 	$f \leftarrow B_i/k$ \\ 
				\REPEATN{$\log(1/\delta)$} \label{line:Mech1L-meyerSTART}
				\STATE $S$: output of Algorithm \ref*{alg:Meyerson-topL} with {$B = B_i$}. 
				\IF{$|S| \le  104k$}
				\STATE $\mathcal{S} \leftarrow \mathcal{S} \cup \{S\}$;
                                compute $d(\C,S)$ using one query per agent
				\ENDIF 
				\ENDREPEAT \label{line:Mech1L-meyerEND}
				\ENDFOR  \label{line:Mech1LForEnd} \label{line:mech1L-inexp1-end} 
                                \STATE If $\Sc=\es$, {\bf return} failure. 
                                Otherwise, let $\bS \assign \arg\min\limits_{S\in\mathcal{S}}\topl(d(\mathcal{C}, S))$. 
                                For $i\in\bS$, set $w_i=\bigl|\{j\in\C: \topalt_{\bS}(j)=i\}\bigr|$; 
                                for all $i\notin\bS$, set $w_i=0$.
                                \label{line:mech1L-inexp2-start} 
				\RETURN Mechanism~\bbtopl$(\bS, \sigma,\{w_j\}_{j\in\bS}, B',\mathcal{A})$ \label{line:Mech1LCalltoBlackBox} \label{line:mech1L-inexp2-end} 
			\end{algorithmic}
			
		\end{namedmech}
		
\begin{theorem}\label{thm:mechOneL} \label{bb-meyerson-thm}
Mechanism~\ref{alg:mech1L} has $O((\log(1/\delta) + \log k )\log n)$ per-agent query
complexity, and achieves $O(1)$-distortion for the $\ell$-centrum problem
with probability at least $1 - \delta$.   
\end{theorem}

The proof of Theorem~\ref{bb-meyerson-thm} relies on Lemma~\ref{sparsify}, which shows that
moving to the weighted instance induced by an $O(1)$-approximate solution (as done in
step~\ref{line:mech1L-inexp2-start} above) results in only an $O(1)$-factor loss. 

	\begin{lemma}\label{lem:wtdToOrig} \label{sparsify}
		Let $S \subseteq \mathcal{C}$ be such that $\topl(d(\mathcal{C},S))\leq\al\cdot\OPT$. 
		The weighted instance induced by $S$ has weights $w_i$, where $w_i=0$ if $i\notin S$, and 
		otherwise $w_i=\bigl|\{j\in\C: i=\topalt_S(j)\}\bigr|$ is the number of points in $\C$ for
		which $i$ is the top choice in $S$.
		Let $\opt'$ be the optimal value of the $\ell$-centrum problem for the weighted instance
		induced by $S$. 
		Then,  
\begin{enumerate}[label=(\alph*), topsep=0.1ex, itemsep=0.1ex, leftmargin=*]
\item $\opt' \le  2( \alpha + 1) \opt$, 
\item If $T$ is a $\rho$-approximate solution with respect to the weighted instance, then 
we have \linebreak 
\mbox{$\topl(d(\mathcal{C}, T)) \le  (\alpha + 2\rho( \alpha + 1)) \cdot \opt$}.
\end{enumerate}	
		\end{lemma}
		
\begin{proof} 
For part (a), let $T^*$ be an optimal solution for the original
		instance. 
Let $\tdT$
		be the projection of $T^*$ onto $S$, that is, the centers obtained by mapping each point
		in $T^*$ to the closest center in $S$. We show an upper bound on $\topl(d(\mathcal{C},
		\tdT | w))$, the $\topl$-cost of the weighted instance with respect to
                $\tdT$. 
                Consider any subset of $\ell$ points, $Q$ (where we take the weights into
		consideration, i.e., we take some $w'_i$ points from each $i \in S$, where $\sum_{i \in S}
		w'_i = \ell$). 	 
			
		For each $i \in Q$, let $x_S(i)$ be the point that $i$ is co-located with in the weighted instance, and $x^*(i)$ be the center in $T^*$ that is closest to $i$. By the triangle inequality, 
		
		$$\sum_{i \in Q} d(x_S(i), \tdT) \le  \sum_{i \in Q} d(x_S(i), i) + \sum_{i \in Q} d(i, x^*(i)) + \sum_{i \in Q} d(x^*(i), \tdT)$$ 
		
		The first term, $\sum_{i \in Q} d(x_S(i), i)$, is the cost incurred when we move each $i \in Q$ from $x_S(i)$ to its original location; this is at most $\topl(d(\mathcal{C}, S))$. The second term, $\sum_{i \in Q} d(i, x^*(i))$, is the cost of moving each $i \in Q$ from its original location to $x^*(i)$, its closest center in $T^*$; the cost of this step is at most $\opt$. Finally, $\sum_{i \in Q} d(x^*(i), \tdT)$ is the cost of moving the points their centers in $T^*$ to their closest open centers in $\tdT$. The cost of this step can be bounded by moving each relevant point in $T^*$ to $\tdT$ -- so we incur an additional cost of at most $\opt + \topl(d(\mathcal{C}, S))$. Putting this together, we have

\begin{equation*}
\begin{split}
  \sum_{i \in Q} d(x_S(i), \tdT) &\le  \sum_{i \in Q} d(x_S(i), i) + \sum_{i \in Q} d(i, x^*(i))
  + \sum_{i \in Q} d(x^*(i), \tdT)\\
  &\le  2\cdot\topl(d(\mathcal{C}, S)) + 2\cdot\opt.
\end{split}
\end{equation*}
		\noindent As this holds for any $\ell$-subset $Q$, $\topl(d(\mathcal{C}, \tdT)) \le  2(\opt + \topl(d(\mathcal{C}, S))) \le  2(\alpha + 1)\opt$. 
		
		It remains to prove that (b) holds. For any solution, $T$, of $\topl$ cost $Z$ for the
		weighted instance, the cost of $T$ for the original instance is at most $Z +
		\topl(d(\mathcal{C}, S))$ (this is an upper bound on the cost of moving the $\ell$
		weighted points to their original locations). Since $\opt' \le  2(\alpha+ 1)\opt$, for any
		$\rho$-approximate solution $T$ for the weighted instance, $\topl(d(\mathcal{C}, T)) \le
		(\alpha + 2\rho(\alpha + 1))\opt$.  
		\end{proof}

\begin{proof}[Proof of Theorem \ref{thm:mechOneL}] 
	Let $\varepsilon \in (0, 1]$ be a
		constant.  Notice that in lines
		\ref*{line:Mech1L-meyerSTART}-\ref*{line:Mech1L-meyerEND}, we are running Algorithm
		\ref{alg:Meyerson-topL} $\log(1/\delta)$ times for a given $B_i$. 
		Since we know that $\opt \le  B' \le  n^2 \cdot \opt$,
		there exists some $i^* \in \{0, \ldots, \left\lceil \log_2n^2 \right\rceil\}$ such
		that $\opt \le  B_{i^*} \le  2 \cdot \opt$.  
			
	We show that with probability at least $1 - \delta$, one of the solutions
	returned by Algorithm~\ref{alg:Meyerson-topL} when $f = B_{i^*}/k$ 
	is a $(104,176)$-bicriteria solution, i.e., it
	opens at most $104k$
	centers, and induces a total connection cost of at most $176\opt(d)$. 
	By Theorem~\ref{thm:Meyerson-lCentrum} and Markov's inequality, the $\topl$ cost of the
	output of Algorithm \ref{alg:Meyerson-topL} when $f = B_{i^*}/k$ is at most 
	$4 \cdot 44 \opt$ with probability at least $\frac{3}{4}$; and the number of centers
	opened is at most $4 \cdot 26k = 104k$, with probability at least $\frac{3}{4}$. Hence, 
	the probability that both events happen is 
	at least $\frac{1}{2}$. 
	Since we run this algorithm
	$\log(1/\delta)$ times, with probability at least $1 - \delta$, there exists 
	$S \in\mathcal{S}$ that is a $(104,176)$-bicriteria solution for $\ell$-centrum.
	It follows that with probability at least $1-\dt$, the solution $\bS$ obtained in
	line~\ref*{line:mech1L-inexp2-start} is a $(104,176)$-bicriteria solution.
	
	Lemma~\ref{sparsify} then shows that moving to the weighted instance induced by $\bS$
	incurs an $O(1)$-factor loss in solution quality. More precisely, let $\OPT'$ denote the
	optimal value of the $\ell$-centrum problem on the weighted instance induced by $\bS$.
	By Lemma \ref{lem:wtdToOrig}, we have $\opt' \le 2(176 + 1)\opt$; also, a good solution to 
	the weighted instance yields a good solution to the original instance.
			
	We would now like to apply the black-box reduction (Mechanism~\bbtopl) to this
	sparsified instance. But one issue is that $\OPT'$ could be much smaller than $\OPT$, and
	so while we do have $\OPT\leq B'\leq n^2\OPT$, we cannot say that $B'$ provides any
	estimate of $\OPT'$. The solution is to utilize the slightly more general guarantee stated
	in Theorem~\ref{thm:blackboxRedL2}. 
	%
	If we take $U =354\opt$, then we have $U\geq\OPT'$, and $U \le  354 B' \le  n^2 \cdot U$,
	and hence we can apply Theorem~\ref{thm:blackboxRedL2} taking $B=354 B'$. 
	So for the weighted instance induced by $\bS$, we obtain a solution of $\topl$-cost at most 
	$354(\rho(1 + 2 \varepsilon) + \varepsilon)\cdot\OPT$.
	Recall that $\rho$ is the approximation factor of the given algorithm $\A$ for
	$\ell$-centrum. 
	By Lemma~\ref{sparsify}, this yields a solution of cost at most
	$\bigl(176+2\cdot 354(\rho(1+2\ve)+\ve)\cdot 177\bigr)\cdot\OPT$ for the original
	instance. 
	In particular,
	taking $\rho=(5 + \varepsilon)$~\cite{ChakrabartyS19},  
	approximation algorithm for $\ell$-centrum given by 
	we obtain a solution of cost at most $O(1)\cdot\OPT$.
		
	\medskip		
	\noindent \emph{Query Complexity}: The total number of per-agent queries made by calls
	to Algorithm~\ref{alg:Meyerson-topL}, and required to compute the costs of the solutions
	added to $\Sc$ is $O(\log(1/\delta) \cdot \log(n))$. 
	Finally, since the weighted instance
	given as input to Mechanism~\bbtopl 
	in line
	\ref*{line:Mech1LCalltoBlackBox} consists of $O(k)$ points,  $B \in [\opt, n^2 \cdot
	  \opt]$, this step takes at most $O(\log n \log k )$ queries per agent (by Theorem
	\ref{thm:blackboxRedL2}).  
	\end{proof}
	
	We remark that while the approximation factor obtained above is quite large, we have not attempted to optimize this at all, and instead chosen simplicity of exposition. Also, it is possible to significantly reduce the approximation factor by using core-set ideas.

\subsection{Adaptive sampling: per-agent query bounds independent of \boldmath $n$} 
\label{sec:adaptiveSample} \label{sec:adsample}

We now develop mechanisms with per-agent query complexity {\em independent} of $n$.
The core algorithmic idea here is {\em adaptive sampling}~\cite{AggarwalDK09,ArthurV07},
which is the following natural idea:   
we successively choose centers, choosing the next center to add to 
the current center-set $S$ by sampling a point $i\in\C$ with probability proportional to
$d(i,S)$. 
In Mechanism~\ref{kmedian}, we do this for $k$ iterations, and~\cite{ArthurV07} showed
that this yields an $O(\ln k)$-approximate $k$-median solution.
Aggarwal et al.~\cite{AggarwalDK09} showed that if we choose $O(k)$ centers this way, then
we obtain an $O(1)$-approximate $k$-median solution (albeit opening $O(k)$ centers)
with high probability.  As described, this fails badly for $\ell$-centrum, indeed even for
$k$-center.  

\begin{theorem}\label{thm:dSample-badEx} \label{adsample-toplbad}
	For any constants $\tau \geq 1$, $L > 1$, $\epsilon > 0$, there exists an instance $\mathcal{I} = (\mathcal{C}, d, k)$ such that $\Pr[ \text{Top}_1(d(\mathcal{C}, S)) < L \cdot \opt] < 2\epsilon$, where $S$ is the set of centers $\tau k$ opened by running Aggarwal et. al's adaptive sampling algorithm on $\mathcal{I}$ and $\opt$ is the value of an optimal $k$-center solution for the instance $\mathcal{I}$. 
\end{theorem}
\begin{proof}		
	Let the set of agents be $\mathcal{C} = C_1 \cup \{j^*\}$, where $|C_1| >    2\tau  + \frac{1}{\epsilon} \cdot 2\tau L $. For all $i, j \in C_1$, $d(i, j) = 1$, and for all $j \in C_1$, $d(j, j^*) = L$. Notice that this defines a valid metric. Fix $k = 2$. An optimal solution for $2$-center would be to open one center in $C_1$, and one center at $j^*$; this solution has a cost of 1, so $\opt = 1$. For any $S \subseteq \mathcal{C}$, if $\text{Top}_1(d(\mathcal{C}, S)) < L = L \cdot \opt$, then $d(j^*, S) < L$; but since $d(i, j^*) = L$ for all $i \in \mathcal{C} \setminus \{j^*\}$, this is only possible if $j^* \in S$. 
	
	Let $S_{i-1}$ be the set of centers opened by the end of step $i-1$ of the d-sampling algorithm, and let $s_i$ be the center opened in step $i$. $\Pr[s_i = j^* | j^* \notin S_{i-1}] = \frac{L}{|C_1| - |S_{i-1}| + L} \leq \frac{L}{|C_1| - 2\tau + L} < \frac{\epsilon}{2\tau}$. By Union bound, $\Pr[ j^* \in  S | j^* \notin S_1] < |S| \cdot \frac{\epsilon}{2\tau} = \epsilon  $. Assuming that the first center is chosen uniformly at random, $\Pr[j^* \notin S_1] = \frac{n-1}{n}$, where $n = |\mathcal{C}|$, so $\Pr[ j^* \in S] = \Pr[ j^* \in S_1] + \Pr[j^* \in S | j^* \notin S_1]\cdot \Pr[ j^* \notin S_1] < \frac{1}{n} +  \epsilon < 2\epsilon$.  Hence, $\Pr[ \text{Top}_1(d(\mathcal{C}, S)) < L \cdot \opt]  \leq \Pr[j^* \in S] < 2\epsilon$.
\end{proof}

Nevertheless, we show how to extend adaptive sampling in a novel fashion for the
$\ell$-centrum problem.  
Again, the insight is that we can exploit the separable proxy function suggested by
Claim~\ref{csproxy}. 
Intuitively, 
adaptive sampling works for $k$-median because, 
given the current set of centers $S$,
we sample the next point to be added to $S$ with 
probability proportional to its contribution to the objective, thereby biasing the
sampling process towards points that currently incur large cost.
The contribution of an agent $i$ to the proxy function given by Claim~\ref{csproxy}
is $\bigl(d(i,S)-t\bigr)^+$, 
which suggests that we should sample a
point $i$ with probability proportional to this. 
(Observe that adaptive sampling for $k$-median corresponds to the special case where
$t=0$.)
We show that this does work:
for a suitable choice of $t$, if we choose
$O(k)$ centers this way (see Algorithm~\ref{adsample-topl}), then we obtain an
$O(1)$-approximate $\ell$-centrum solution with 
high probability, nicely generalizing the guarantee of (standard) adaptive sampling for
$k$-median. 
In the analysis, we need various new ideas to deal with the fact that distances of the
form $\bigl(d(i,j)-t\bigr)^+$ do not form a metric.
  
		\begin{namedalg}[h!]{\adsampltopl}
			\caption{\hfill Adaptive sampling for $\ell$-centrum} 
                        \label{alg:Lcentrum} \label{adsample-topl}
			\textbf{Input}: instance $(\mathcal{C}, d)$,
                        parameter $t_\ell\geq 0$
			\begin{algorithmic}[1]
				\STATE $S_0 \leftarrow \varnothing$ \; 
				\FOR{$i=1,\ldots,\ceil{28(k + \sqrt{k})}$}
				\STATE Sample $s_i$ with probability proportional to
                                $\bigl(d(s_i, S_{i-1}) - 2t_\ell\bigr)^+$ \\  
				\STATE Update $S_i \leftarrow S_{i-1} \cup \{s_i\}$. 
				\ENDFOR 
				\RETURN $S_{\ceil{28(k + \sqrt{k})}}$\\
			\end{algorithmic}
\begin{remark}
We have assumed above that the metric $d$ is given. If we are only given a
preference profile, then in each iteration, we make one value query to each agent 
$j\notin S_{i-1}$ to compute $d\bigl(j,\topalt_{S_{i-1}}(j)\bigr)$, and thus implement the 
sampling procedure. 
The resulting mechanism has $O(k)$ per-agent query complexity.
\end{remark}
\vspace*{-1.25ex}
\end{namedalg}

Fix some optimal solution, and let $t^*_\ell$ be the $\ell$-th largest distance
between any voter and their preferred candidate in this solution. 
When the parameter $t_\ell$ is sufficiently close to $t^*_\ell$, we have the following approximation guarantee for Algorithm \ref{alg:Lcentrum}. 

\begin{theorem}\label{thm:AdaptiveSample-lcentrum} \label{adsample-topl-thm}
\label{thm:AdaptiveSample-lcentrum-orig}
Let $t_\ell$ be such that
$t^*_\ell\leq t_\ell\leq\max\{(1+\ve)t^*_\ell,\frac{\ve\OPT}{\ell}\}$, for some $\ve>0$. 
Algorithm \ref{alg:Lcentrum} run with parameter 
$t_\ell$ opens at most $56k$ centers, 
and returns a solution of $\topl$-cost at most $35(1+\ve)\cdot\OPT$ with constant
probability.    
\end{theorem}	

To avoid detracting the reader, we defer the proof of Theorem \ref{adsample-topl-thm} to
Section~\ref{adsample-topl-proof}. To compute a suitable $t_\ell$ (satisfying the
conditions of Theorem \ref{adsample-topl-thm}), we utilize the estimates $B_1$ and $B_n$
described in Theorem~\ref{kcenter-thm} to compute a small set of guesses that contains a
suitable choice of $t_\ell$.  
Fix $\varepsilon > 0$ in the sequel. 

\begin{claim}\label{claim:setOfGuesses}
Let $B_1$ and $B_n$ be estimates given by Theorem~\ref{kcenter-thm}.
Define  $\mathcal{T}_1 = \{
B_1 \cdot (1+\varepsilon )^{-r} : r = 0, \ldots, \log_{1+\varepsilon
}(\frac{2\ell^2}{\varepsilon} )\}$ and 
$\mathcal{T}_{2} = \{ B_n \cdot (1+\varepsilon )^{-r}
: r = 0, \ldots, \log_{1+\varepsilon }(\frac{(8\ln(k) + 4) \cdot n}{\varepsilon}
)\}$.  There are $t'_\ell \in \mathcal{T}_1$,  $t''_\ell \in \mathcal{T}_2$ such that
$t^*_\ell \le  t'_\ell, t''_\ell  \le  \max\{(1+\varepsilon )t^*_\ell, \varepsilon  \cdot
\frac{\opt}{\ell}\}$. 
\end{claim}

\begin{proof} 
	Recall $\mathcal{T}_1 = \{ B_1 \cdot (1+\varepsilon )^{-r} : r = 0, \ldots,
	\log_{1+\varepsilon }(\frac{2\ell^2}{\varepsilon} )\}$.  
	Recall that $\opt \leq B_1\leq 2\ell \opt$, so $t^*_\ell \le  \opt \le
	B_1$ and $\ve\cdot\frac{B_1}{2\ell^2}\le \varepsilon  \cdot \frac{\opt}{\ell}$.  
	Hence, there exists $t_\ell \in \mathcal{T}_1$ such that $t^*_\ell \le t_\ell \le  \max\{(1+\varepsilon
	)t^*_\ell, \varepsilon  \cdot \frac{\opt}{\ell}\}$. 
	
	The other set is $\mathcal{T}_{2} = \{  B_n \cdot (1+\varepsilon )^{-r} : r = 0, \ldots,
	\log_{1+\varepsilon }(\frac{(8\ln(k) + 4) \cdot n}{\varepsilon} )\}$. Recall that
	$\opt \leq B_n \leq (8 \ln (k) + 4)\cdot \frac{n}{\ell} \cdot \opt$.   Note that
	$t^*_\ell \leq B_n$, and $\ve\cdot\frac{B_n}{(8\ln(k) + 4)
	  n}\leq\ve\cdot\frac{\opt}{\ell}$, 
	so there exists $t_\ell \in \mathcal{T}_2$ such that $t^*_\ell \le t_\ell \le
	\max\{(1+\varepsilon )t^*_\ell, \varepsilon  \cdot \frac{\opt}{\ell}\}$.  
	\end{proof}

Combining Algorithm~\ref{adsample-topl}
with the set of guesses for $t^*_\ell$ prescribed by Claim~\ref{claim:setOfGuesses} yields
Mechanism~\ref{adsample}, stated below.  
		
		\begin{namedmech}[h!]{\adsamplagent}
                  \caption{\hfill $\tO\bigl(k\log(\min\{\ell,n/\ell\})\bigr)$ per-agent query
                    complexity}
			\label{alg:mech2L} \label{adsample}
			\textbf{Input}: {Preference profile $\sigma$, $\rho$-approximation
                          algorithm $\mathcal{A}$ for $\ell$-centrum, where $\rho=O(1)$}
			\begin{algorithmic}[1]
				\STATE $\mathcal{T} \leftarrow \arg\min\{ |\mathcal{T}_1|,
                                |\mathcal{T}_2|\}$, 
                                where $\T_1,\T_2$ are from Claim~\ref{claim:setOfGuesses}, 
                                \quad $\mathcal{S} \leftarrow \varnothing$ \label{line:mech2l-inexp1-start}
				\FOR{each $t_\ell \in \mathcal{T}$, repeat
                                  $\log(1/\delta)$ times}
                                \label{line:Mech2LSamplingStart}
				\STATE $S$: output of Algorithm \ref*{alg:Lcentrum} using
                                parameter $t_\ell$ \\ 
				\STATE $\mathcal{S} \leftarrow \mathcal{S} \cup \{S\}$;
                                compute $d(\C,S)$ using one query per agent
				\label{line:Mech2LSamplingEnd} 
				\ENDFOR \label{line:mech2l-inexp1-end}
				
				\STATE Let $\bS\assign\arg\min\limits_{S\in\mathcal{S}}\topl(d(\C,S))$  
                                \label{line:mech2l-inexp2-start}
                                \STATE Query $d(j,a)$ for all $j\in\C$, $a\in\bS$  
                                \RETURN $\mathcal{A}\bigl((\C,\bS),d\bigr)$
			\end{algorithmic}
		\end{namedmech}
		
\begin{theorem}\label{thm:mechTwoL} \label{adsample-thm}
Mechanism~\ref{alg:mech2L} has $\tO\bigl(k \log(1/\delta) \log(\min\{\ell,n/\ell\})\bigr)$
per-agent query complexity,  
and achieves $O(1)$ distortion for the $\ell$-centrum problem with
probability at least $1 - \delta$.   
\end{theorem}

\begin{proof} 
	In order to compute $\mathcal{T}$, we require estimates of $\opt_\ell$, $B_1$ and $B_n$, satisfying the conditions of Claim \ref{claim:setOfGuesses}. By Theorem \ref{kcenter-thm}, we can compute such a $B_1$ and $B_n$ using Mechanisms \ref{alg:kcenter} and \ref{alg:AV} respectively. 							
	Let $t_\ell \in \mathcal{T}$ be such that $t^*_\ell \le t_\ell \le  \max\{(1+\ve)t^*_\ell,
	\ve\cdot  \frac{\opt}{\ell}\}$. we run Algorithm \ref{alg:Lcentrum} $\log(1/\delta)$
	times.  By Theorem \ref{thm:AdaptiveSample-lcentrum}, if we run Algorithm
	\ref{alg:Lcentrum} $O\bigl(\log(1/\delta)\bigr)$ times with $t_\ell$ 
	we will
	obtain a $\bigl(56,35(1+\ve)\bigr)$-bicriteria solution for $\ell$-centrum,
	with probability at least $1 - \delta$. 
	Hence, with probability $1-\dt$, $\bS$ is such a bicriteria solution.

We construct the entire metric on
the instance $\C\times\bS$, where $\C$ is the client-set and $\bS$ is the facility-set, so
we can run the $(5+\ve)$-approximation algorithm of~\cite{ChakrabartyS19} on this
instance. 
Let $T\sse\bS$ be the $\ell$-centrum solution returned. We argue that $T$ is a good
$\ell$-centrum solution for the original instance as well.
Let $\OPT_{\C\times\bS}$ denote the optimal $\ell$-centrum value for the $\C\times\bS$
instance. We have $\OPT_{\C\times\bS}\leq 2\OPT+\topl(d(\C,\bS))$. This is because we
can take an optimal solution $S^*\sse\F$ for the original instance, and map each 
$a\in S^*$ to the center $a'\in\bS$ minimizing $\min_{s\in\C}(d(s,a)+d(s,a'))$,
to obtain a center-set $F\sse\bS$. Consider any $j\in\C$.
Let $a=\topalt_{S^*}(j)$, $a$ be mapped to $a'\in F$, and $a''=\topalt_{\bS}(j)$.
We have 
\begin{equation*}
\begin{split}
d(j,F) & \leq d(j,a')\leq d(j,a)+\min_{s\in\C}(d(s,a)+d(s,a'))
\leq d(j,a)+\min_{s\in\C}(d(s,a)+d(s,a'')) \\
& \leq d(j,a)+(d(j,a)+d(j,a'')) 
= 2d(j,S^*)+d(j,\bS).
\end{split}
\end{equation*}
The second inequality is due to the triangle inequality, and the third inequality is because $a$
is mapped to $a'$.
Summing over any set of $\ell$ agents, yields 
$\topl(d(\C,F))\leq 2\OPT+\topl(d(\C,\bS))\leq 37(1+\ve)\OPT$.
Therefore, we have 
$\topl(d(\C,T))\leq\rho\cdot\OPT_{\C\times\bS}\leq 37\rho(1+\ve)\OPT$.
							
	\medskip
	\noindent \emph{Query Complexity}: By Theorem \ref{kcenter-thm}, the per-agent query
	complexity of Mechanism~\ref{alg:kcenter} and Mechanism~\ref{alg:AV} is $k$; hence, computing
	$\mathcal{T}$ only requires $2k$ queries per agent. The size of $\mathcal{T}$, and hence
	the number of $t_\ell$ values considered is 
	$O(\log(\min\{\ell, \ln(k)n/\ell\}))=\tO\bigl(\ln(\min\{\ell,n/\ell)\bigr)$. 
	For each $t_\ell$, Algorithm~\ref{alg:Lcentrum}, which can be implemented using $O(k)$
	per-agent queries, is run $\log(1/\delta)$ times. Finally, a
	total of $O(k)$ value queries per agent are made when computing pairwise-distances
        $d(j,a)$ for points $j\in\C$ and $a\in\bS$, since $|\bS| = O(k)$.  
	Thus, the total number of queries per agent is 
	$\tO\bigl(k \log(1/\delta)\log(\min\{\ell,n/\ell\}))$.  
	\end{proof}

\subsection{Adaptive sampling: total-query-complexity bounds}
We now devise a mechanism whose total query complexity depends on $\polylog(n)$, which is vastly better
than the linear dependence on $n$ 
that follows 
from Mechanisms~\ref{bb-meyerson} or~\ref{adsample}. 
To obtain this, 
we change how the
adaptive-sampling is implemented in Algorithm~\ref{adsample-topl}. 
Instead of querying agents outside of the current-center set $S$ to obtain 
$d(\C,S)$,
we now construct this vector approximately by querying agents in $S$.
Similar to our black-box reduction, we consider a distance threshold $\thresh$, and find
the ring of points $a\in\C$ for which $d(a,S)\in(\thresh,(1+\ve)\thresh]$.
This can be done via binary search on $j$'s preference profile, for each $j\in S$.
We consider geometrically increasing thresholds, using the estimate $B_1$ obtained from
the $k$-center mechanism 
to hone in on a $\poly(n)$-bounded 
range of distance thresholds. Thus, we need to consider $O(\log n)$ $\thresh$ values, and
so the total number of queries involved is $O(|S|\log^2 n)$. Since $d(a,S)$ is roughly the
same for all points in a ring, we sample by first choosing a ring, and then a uniform
point in the ring.  
With this ring-based implementation of adaptive sampling (Algorithm ~\ref{adsample-ring}),
we proceed as in Mechanism~\ref{adsample}, except that we utilize only $B_1$ to obtain the
candidate set $\T$ of $t_\ell$ values since this can be computed using $O(k^2)$ queries in
total (Theorem~\ref{kcenter-thm} (a)). 
The resulting mechanism  has total-query-complexity 
$O(k^2\log^2 n\log\ell)$. 

%
                

\begin{namedalg}[h!]{\adsamplring}
\caption{\hfill Ring-based adaptive sampling}
			\label{alg:Lcentrum:ring} \label{adsample-ring}
			\textbf{Input}: {$\ell$-centrum instance $(\C, d)$, parameters $t_\ell,\ve$}
\begin{algorithmic}[1]
   \STATE {$(S_0, B)$: output of Mechanism \ref{mech:kcenter}} \label{kcentout}
      \FOR{$i = 1, \ldots,124k$}
      \STATE For $h = 0, \ldots, N:=\log(2n^2/\ve)$, define thresholds $\thresh_h = \frac{B}{2^{N - h}}$  
      \STATE Partition $\C\sm S_{i-1}$ into rings $R_{\thresh_0}, \ldots, R_{\tau_N}$, where
      $R_{\thresh_h} =  \{j \notin S_{i-1} : d(j, S_{i-1}) \in(\thresh_h/2,\thresh_h]\}$
        if $h \in [N]$ and,    
        $R_{\thresh_0} = \{j \notin S_{i-1} : d(j, S_{i-1})\leq\thresh_{0}\}$.
       \STATE Sample exactly one index in $\{0,\ldots,N\}$, choosing index $h$
          with probability proportional to 
          $|\ring_{\thresh_h}|(\thresh_h - \rbeta t_\ell)^+$. 
          Choose $s_i$ uniformly at random from $\ring_{\thresh_h}$.
          \label{adsample-ring-samppt}
          \STATE Set $S_i\assign S_{i-1} \cup s_i$
          \ENDFOR 
          \RETURN $S_{124k}$
\end{algorithmic}

\begin{remark} \label{adsample-ringremk}
We have assumed above that the metric $d$ is given. If we are only given a
preference profile, then in each iteration, we compute the rings using $O(|S_i|\log^2 n)$
total number of queries.
%
Moreover, we can estimate the $\topl$-cost of $S=S_{124k}$ without
any further queries, 
as follows. 
Find the largest index $j\in\{0,\ldots,N\}$ such that 
$\sum_{r=j}^N|R_{\tau_r}|\geq\ell$, and return 
$\sum_{r=j+1}^N\tau_r\cdot|R_{\tau_r}|+(\ell-\sum_{r=j+1}^N|R_{\tau_r}|)\tau_j$, which
well-estimates $\topl(d(\C,S))$ (see Theorem~\ref{adsample-ring-thm}).
\end{remark}
\vspace*{-1.25ex}
\end{namedalg}

The following result shows that the above ring-based adaptive sampling
indeed yields a
constant-factor bicriteria approximation for the $\ell$-centrum problem. 

\begin{theorem}\label{thm:AdaptiveSample-lcentrum-ring}\label{adsample-ring-thm}
Let $t_\ell$ be such that $t^*_\ell \le  t_\ell \le  \max\{(1+\varepsilon )t^*_\ell, \varepsilon  \cdot
\frac{\opt}{\ell}\}$. 
Algorithm \ref{alg:Lcentrum:ring} opens at most $125k$ centers, and when run with parameter
$t_\ell$, returns a solution $S$ 
having $\topl$-cost at most $50(1+2\ve)\cdot\OPT$ with constant probability. 
Moreover, the estimate of $\topl(d(\C,S))$ computed in Remark~\ref{adsample-ringremk} lies
in the interval $\bigl[\topl(d(\C,S)),(2+\ve)\topl(d(\C,S))\bigr]$.
\end{theorem}

We defer the proof of Theorem~\ref{adsample-ring-thm} to Appendix \ref{app:aeqc}. We
describe here the mechanism obtained using ring-based adaptive sampling, and analyze  
its performance.


		\begin{namedmech}[h!]{\adsampltot} 
			\caption{\hfill $O(k^2 \log^2 n \log \ell )$-total-query complexity}
			\label{alg:mech2L:ring} \label{adsample-tot}
			\textbf{Input}: {Preference profile $\sigma$, $\rho$-approximation
                          algorithm $\mathcal{A}$ for $\ell$-centrum, where $\rho=O(1)$}
			\begin{algorithmic}[1]
   \STATE $S_0, B_1$: output of Mechanism~\ref{mech:kcenter} \\ 
				\STATE    $\mathcal{S} = \varnothing$,  $\mathcal{T}=\{
                                \ell B_1 \cdot (1+\varepsilon )^{-r} : r = 0, \ldots,
                                \log_{1+\varepsilon }(\frac{2\ell^2}{\varepsilon} )\}$,
                                where $0<\ve\leq 1$
				\FOR{each $t_\ell \in \mathcal{T}$}
				\REPEATN{$O\bigl(\log(1/\delta)\bigr)$} 
				\STATE $S$: output of Algorithm~\ref{adsample-ring} using
                                parameters $t_\ell,\ve$ \\ 
				\STATE Estimate $\topl(d(\C,S))$ as described in
                                Remark~\ref{adsample-ringremk} 
				\ENDREPEAT 
				\ENDFOR
				
				\STATE Let $\bS\in\Sc$ be the solution with smallest
                                estimated cost.
                                For $i\in\bS$, set $w_i=\bigl|\{j\in\C: \topalt_{\bS}(j)=i\}\bigr|$; 
                                for all $i\notin\bS$, set $w_i=0$.
				\STATE Query $d(i, j)$ for all $i, j \in\bS$ 
				\RETURN $\mathcal{A}(\bS, w, d)$ 
			\end{algorithmic}
		\end{namedmech}

\FloatBarrier
						
\begin{theorem}\label{thm:mechTwoL:ring}
Mechanism~\ref{alg:mech2L:ring} has $O(k^2\log^2(n)\log(\ell)\log(1/\delta))$
total query complexity, and achieves $O(1)$ distortion for
the $\ell$-centrum problem with probability at least $1 - \delta$.   
\end{theorem}

\begin{proof} 
By Theorem \ref{kcenter-thm}, the output of the $k$-center mechanism (Mechanism \ref{alg:kcenter}), $S_0$, is a
2-approximate $k$-center solution.  
Let {$t_\ell \in \mathcal{T}$ be such that $t^*_\ell \le t_\ell \le  \max\{(1 +
  \varepsilon)t^*_\ell, \varepsilon\cdot  \frac{\opt}{\ell}\}$}. 
By Theorem \ref{thm:AdaptiveSample-lcentrum-ring}, since we run
Algorithm \ref{alg:Lcentrum:ring} $\log(1/\delta)$ times with $t_\ell$, we will
obtain a {$(125,50(1 + 2\varepsilon))$}-bicriteria 
$\ell$-centrum solution, $S$, with probability at least $1 - \delta$;
also, the estimate we compute via Remark~\ref{adsample-ringremk} has value at most
$50(2+\ve)(1+2\ve)\OPT$.  

It follows that $\bS$ is a $(125,50(2+\ve)(1+2\ve))$-bicriteria $\ell$-centrum solution
with probability at least $1-\dt$.
We query all pairwise distances for $i, j \in\bS$ and use the algorithm of
\cite{ChakrabartyS19} to obtain an $O(1)$-approximation for the weighted instance,
and hence for the original instance (due to Lemma~\ref{sparsify}).
							
\medskip
\noindent\emph{Query complexity}.  By Theorem \ref{kcenter-thm}, Mechanism
\ref{alg:kcenter} has $O(k^2)$ total-query complexity. As discussed in
Remark~\ref{adsample-ringremk}, Algorithm~\ref{adsample-ring} can be 
implemented using $O(k^2\log^2 n)$ queries, and we can estimate $\topl(d(\C,S))$ with no
additional queries.
We run Algorithm~\ref{adsample-ring} $O(\log\ell\log(1/\dt))$ times, so the total query
complexity  
of the mechanism is $O(k^2 \log^2(n) \log(\ell) \log(1/\delta))$.  
\end{proof}

\subsection{Analysis of Algorithm~\ref{meyerson-topl}: Proof of Theorem~\ref{thm:Meyerson-lCentrum}} 
\label{append-meyerson-topl} \label{mey-toplproof}

We actually prove a slightly stronger statement, for a generalization of Algorithm
\ref{alg:Meyerson-topL}, which will also allow us to apply it to the setting $\F \neq \C$.
  
\begin{namedalg}[h!]{\meytoplgen} 
    \caption{\hfill Extension of
      Algorithm~\ref{alg:Meyerson-topL}}\label{alg:Meyerson-topL:FC}
    \label{meyerson-toplgen}
    \textbf{Input}: {Sequence of agents $x_1,\ldots,x_n$, estimate $B\geq\opt$, 
      parameter $\topDist \in \{0, 1\}$. }\\
    \vspace{-10pt}
    \begin{algorithmic}[1]
        \STATE  $S \leftarrow \{x_1\}$, $f = \frac{B}{k}$\\
        \FOR{$i=2,\ldots, n$}
        \STATE $\delta_i = \left(d(x_i, S) - (3 + \topDist)\cdot \frac{B}{\ell}\right)^+$\\
        \STATE Add $\topalt(x_i)$ to $S$ with probability $\min(1, \delta_i/f)$\\
        \ENDFOR
        \RETURN {$S$}
    \end{algorithmic}
    
\end{namedalg}
\begin{theorem} \label{thm:Meyerson-lCentrum:FC} \label{meytoplgen-thm}
  Let $S^*$ be an optimal solution to the $\ell$-centrum problem.  If, for some
  $\topDist \in \{0, 1\}$, we have $d(j,\topalt(j)) \leq \topDist\cdot d(j, S^*)$ for all $j \in
  \C$, and the order of agents is random, then the expected number of facilities opened
  by Algorithm \ref{alg:Meyerson-topL:FC} is at most $ (26 + 16\topDist) k$,  and
  the expected cost is at most $(15 + 4\topDist)B + (14 + 13\topDist)\opt$.  
\end{theorem}
Note that we can always take $\topDist=1$ above. But in the setting $\F = \C$, we have
$\topalt(j) = j$ for every agent $j$, so we can take $\topDist=0$;
then Algorithms~\ref{alg:Meyerson-topL}
and~\ref{alg:Meyerson-topL:FC} coincide, and Theorem \ref{thm:Meyerson-lCentrum:FC} 
yields the guarantees stated in Theorem \ref{thm:Meyerson-lCentrum}.   

We now prove Theorem \ref{thm:Meyerson-lCentrum:FC} by suitably adapting Meyerson's proof
for facility location~\cite{Meyerson01}. 
We bound the expected value of $\sum_{j\in\C}\bigl(d(j,S) - (3 + \topDist) t\bigr)^+$, for
$t=B/\ell$, where $B$ is an estimate of the optimal value,
which then also yields a bound on the expected $\topl$-cost (via Claim
\ref{claim:CS19proxy})

Fix an optimal solution $S^* = \{c^*_1, \ldots, c^*_k\} \in \F^k$. Let 
$C^*_1, \ldots, C^*_k$ be the clusters induced by $S^*$; that is, for $q \in [k]$, $C^*_q \subseteq \C$ is the set of agents $j$ assigned to center $c^*_q$. Let $t^*_\ell$ be the $\ell$th largest assignment cost induced by $S^*$; notice that, as $B \geq \opt$, at most $\ell$ agents can have a cost larger than $\frac{B}{\ell}$ under $S^*$, and hence $\frac{B}{\ell} \geq t^*_\ell$. 

We first give an outline of the proof. We consider the expected cost $\sum_{j \in \C} \mathbb{E}[\min(Bk, (j, S) - (3 + \topDist)\frac{B}{\ell})^+)]$. If this cost is $O(B + \sum_{j \in \C} (d(j, S^*) - \frac{B}{\ell})^+$, then since $\frac{B}{\ell} \geq t^*_\ell $, we can infer that $S$ has $\topl$-cost $O(B + \opt)$. 

To bound $\sum_{j \in \C} \mathbb{E}[\min(Bk, (j, S) - (3 + \topDist)\frac{B}{\ell})^+)]$, we follow the approach of \cite{Meyerson01} and  consider the ``core'' and ``non-core'' agents separately (we will define the notion of the core of a cluster shortly). If we restrict our attention to the core-agents only, the expected cost incurred before a core-agent is opened is not large (by Lemma \ref{lem:Liberty}); moreover, once a core-agent is chosen, the expected cost incurred by the other agents in the core can be bounded via the triangle inequality. For each of the remaining (non-core) agents, we can bound the expected cost the agent incurs in terms of last core-agent preceding it, if such an agent exists. If no such agent exists (i.e. the non-core agent precedes \emph{all} core-agents), the incurred cost may be large; fortunately, the probability of this event is small (as the order of the agents is random), and hence the \textit{expected} cost is still sufficiently small in this case.


We now proceed with the details. For ease of exposition, we will define $t_\ell :=
\frac{B}{\ell}$. 
For $q\in[k]$, define the \emph{radius} of $C^*_q$ to be $r_\ell(C^*_q) = \sum_{j \in C^*_q}
\frac{(d(j, c^*_q) - t_\ell)^+}{|C^*_q|}$. The \emph{core} of the cluster is the set of
agents in $C^*_q$ that are close to its center $c^*_q$.  
 \begin{definition}The \emph{$\ell$-core} of a cluster $C^*_q$ is defined as  $\core_\ell(C^*_q) = \{j \in C^*_q: (d(j, c^*_q) - t_\ell)^+ \le 2r_\ell(C^*_q)\}$, where $r_\ell(C^*_q) = \sum_{j \in C^*_q} \frac{(d(j, c^*_q) - t_\ell)^+}{|C^*_q|}$. 
\end{definition}
We state some properties of $\core_\ell(C^*_q)$, that will be of use later. First, by
Markov's inequality, the number of agents in the core is large, at least
$\frac{|C^*_q|}{2}$. Furthermore,  by the triangle inequality, for any $j \in
\core_\ell(C^*_q)$, we have $d(\topalt(j), c^*_q) \leq (1 + \topDist) \cdot d(j, c^*_q)$.

As described earlier, we will bound the expected cost incurred by the agents in
$\core_\ell(C^*_i)$ and not in $\core_\ell(C^*_i)$ separately. 
In general, the probability that a center is opened at a given location is dependent on the sequence in which core-agents are considered (event $\mathcal{E}_1$), the centers opened outside the core (event $\mathcal{E}_2$), and the number of core-agents considered before a center is opened for the first time ($\mathcal{E}_3$). For ease of exposition, we define $\mathcal{E} = \mathcal{E}_1 \wedge \mathcal{E}_2 \wedge \mathcal{E}_3$. 
	
To compute an upper bound on the expected cost incurred by \emph{core-agents} we bound the incurred cost before and after an agent in $\core_\ell(C^*_i)$ is selected. To bound the cost incurred before a center at a core-agent is opened, we will use the following lemma.  

\begin{lemma}[Liberty et al.~\cite{LibertySS16}]\label{lem:Liberty}	
  Let $X_1, \ldots X_n$ be a sequence of $n$ independent experiments, where each experiment
  succeeds with probability $p_i \geq \min \{A_i / B, 1\}$ where $B \ge 0$ and $A_i \ge
  0$ for all $i=1,\ldots,n$. Let $t$ be the (random) number of consecutive unsuccessful
  experiments before the first successful one. Then,  $\E{\sum_{i=1}^{t} A_i} \le B$. 
	\end{lemma}

Fix a cluster $C^*=C^*_i$, where $i\in[k]$, and let $c^*$ be its center.
	We begin by bounding $\sum_{j \in\core_\ell(C^*)} \mathbb{E}[\min(\delta_j,
          f)]$. Let $g_1, \ldots, g_{j^*}, \ldots, g_q$ be the core-agents in $C^*$ (in
        the order that they are considered by Algorithm \ref{alg:Meyerson-topL:FC}), where
        $q = |\lcore|  \ge  \frac{|C^*|}{2}$, and $g_{j^*}$ is the first core-agent at
        which a center is opened. Once a center has been opened at $\topalt(g_{j^*})$, for
        any subsequent core-agent $g_i$, $\delta_{g_i} \leq (d(g_i, \topalt(g_{j^*}) - 3
        t_\ell)^+ \leq (d(g_i, c^*) - t_\ell)^+ + (d(\topalt(g_{j^*}), c^*) - 2t_\ell)^+$ (by
        the triangle inequality, and since $(y + z)^+ \leq y^+ + z^+$). Since $g_i,
        g_{j^*} \in \lcore$, this quantity is at most $ (d(g_i, c^*) - t_\ell)^+ + 2(1 +
        \topDist) \rl$.  
 
 It remains to bound $\mathbb{E}[\min(\delta_g, f)]$ for core-agents $g$ that precede $g_{j^*}$. The events of opening centers at core-agents preceding $g_{j^*}$ are independent when we condition on $\mathcal{E}$ (the sequence in which core-agents are considered, the centers opened outside the core, and the number of core-agents considered before a center is opened for the first time). Hence, by Lemma \ref{lem:Liberty}, the expected value of $\sum_{i=1}^{j^* - 1} \min(\delta_{g_i}, f)$, when conditioned on $\mathcal{E}$, is at most $f$.  Thus, we obtain the following bound on the total expected cost (conditioned on $\mathcal{E}$): 
\begin{alignat}{1}
\Exp\Bigl[\sum_{g \in\core_\ell(C^*)}\hspace*{-1ex} \min(\delta_{g}, f)\, |\, \mathcal{E}\Bigr]
&\leq   f+ \E{\min(\delta_g, f)\, |\, \mathcal{E}}  +    \sum_{i=j^* + 1}^q   (d(g, c^* ) - t_\ell)^+ 
+  2(1 + \topDist)|\lcore| \cdot r_\ell(C^*) \notag \\
&\leq 2f   +    \sum_{i=j^* + 1}^q   (d(g, c^* ) - t_\ell)^+  +  2(1 + \topDist) |\lcore| \cdot r_\ell(C^*) 
\label{eq:GoodBound:FC}
\end{alignat}
 
	We now bound $\mathbb{E}[\min(\delta_b, f)]$ for a non-core agent $b \in C^* \setminus \lcore$, in terms of the expected cost of agents in $\lcore$ that precede it. We will use $\prev(b)$ to denote the last agent in $g_1, \ldots, g_q$ that precedes $b$ (if no such agent exists, $\prev(b) = \varnothing$).  

 First, if $\prev(b) = \varnothing$ (i.e. $b$ precedes all core agents), we simply bound $\min(\delta_b, f)$ by $f$. Since the ordering of agents is uniformly random, this event happens with a probability of $\frac{1}{q + 1} \le \frac{2}{|C^*_i|}$ (where $q = |\lcore|$). 
	
Suppose $\prev(b) = g_i$. Let $S_{g_i}$ be the set of centers that are open immediately \emph{after} $g_i$ is considered. By the triangle inequality, $\delta_b \le  (d(b, S_{g_i}) - (3 + \topDist) t_\ell)^+ \leq (d(b, c^*) - t_\ell)^+ + (d(c^*, g_i) - t_\ell)^+ + (d(g_i, S_{g_i}) - t_\ell)^+$. Moreover, as $g_i \in \lcore$, $d(g_i, c^*) \leq 2 \rl$.
      We consider two cases here: 	

\begin{enumerate}[label=--, topsep=0.5ex, itemsep=0.2ex, leftmargin=*]
		\item If $g_i$ is close to the set of open centers, particularly if $d(g_i, S_{g_i}) \le (4 + \topDist)t_\ell$, $\delta_b \le (d(b, c^*) - t_\ell)^+ + 2\rl + (3 + \topDist)t_\ell$.
  
  \item Otherwise,  $d(g_i, S_{g_i}) > (4 + \topDist)t_\ell$. It is easy to see that  $(d(g, S_g) - t_\ell)^+ \le  (3 + \topDist)(d(g, S_g) - (3 + \topDist)t_\ell)^+$, and hence,
		\begin{equation*}
                  \begin{split}
			\delta_b &\le (d(b, c^*) - t_\ell)^+ + 2\rl
            + (3 + \topDist)(d(g, S_g) - (3 + \topDist)t_\ell)^+  \\
			&= (d(b, c^*_i) - t_\ell)^+ +  2r_\ell(C^*)  + (3 + \topDist)\delta_{g_i}
		  \end{split}
                \end{equation*}
		Since $g_i$ is far away from $S_{g_i}$, no center was opened at $\topalt(g_i)$, and hence, $\min\{\delta_{g_i}, f\} = \delta_g$. So, in this case,  $\delta_b \le  (d(b, c^*_i) - t_\ell)^+ +  2r_\ell(C^*)  + (3 + \topDist)\min(\delta_g, f)$.
	\end{enumerate}

Thus, the expected value of $\min(\delta_b, f)$, conditioned on $\mathcal{E}$, is at most 
\begin{alignat*}{1}
\Pr[\prev(b) = \varnothing] \cdot f & + \sum_{i=1}^q \Pr[\prev(b) = g_i] \cdot (d(b, c^*) -t_\ell)^+ \\
&   + \sum_{i=1}^q \Pr[\prev(b) = g_i] \cdot  ( 2\rl + (3 + \topDist)(\min(\delta_{g_i}, f) + t_\ell)). 
\end{alignat*}
Since $\Pr[\prev(b) = g] = \frac{1}{q + 1} \leq \frac{2}{|C^*|}$ for any $g \in \{g_1, \ldots, g_q\} \cup \{\varnothing\}$, this bound can be further simplified to 
\[ \frac{f + (d(b, c^*) - t_\ell)^+ + 2\rl + (3 + \topDist)(t_\ell + \sum_{i=1}^q \min(\delta_{g_i}, f))}{|C^*|/2}.  \]

 By summing this bound over all non-core agents in $C^*$, we obtain the following bound 
 \begin{equation*}
\begin{split}
\Exp\Bigl[\sum_{b \in C^* \setminus \lcore}\hspace*{-1ex} \min(\delta_b, f)\, |\,\mathcal{E}\Bigr] & \leq  
f + \sum_{b \in C^* \setminus \lcore}  (d(b, c^*) - t_\ell)^+
+ (3 + \topDist)\sum_{j=1}^q \E{\min(\delta_{g_j}, f)\, |\, \mathcal{E}}  \\
& + |C^* \setminus \lcore| \cdot (2\rl + (3 + \topDist)t_\ell). 
\end{split}
\end{equation*}

We can combine this with the earlier bound \eqref{eq:GoodBound:FC} for core-agents to
obtain that  
$\E{\sum_{j \in C^*} \min(\delta_j, f)\, |\, \mathcal{E}}$ is at most 
\begin{alignat}{1}
& f + \sum_{b \in C^* \setminus \lcore}\hspace*{-1ex}  (d(b, c^*) - t_\ell)^+  
+ |C^* \setminus \lcore| \cdot (2\rl + (3 + \topDist)t_\ell) 
+ (4 + \topDist)\sum_{j=1}^q \mathbb{E}[\min(\delta_{g_j}, f)\, |\, \mathcal{E}] \notag
\\
& \leq f + \sum_{b \in C^* \setminus \lcore}\hspace*{-1ex}  (d(b, c^*) - t_\ell)^+  
+ |C^* \setminus \lcore| \cdot (2\rl + (3 + \topDist)t_\ell) \notag \\
& \qquad \qquad \quad +(4 + \topDist)\Bigl[
2f   +    \sum_{i=j^* + 1}^q   (d(g, c^* ) - t_\ell)^+  +  2(1 + \topDist) |\lcore| \cdot r_\ell(C^*)\Bigr]
\notag \\
&\leq (9 + 2\topDist)f + (4 + \topDist)\sum_{j \in C^*} (d(j, c^*) - t_\ell)^+ 
+ (3 + \topDist)|C^* \setminus \lcore| \cdot t_\ell  + (10 + 12\topDist) |C^*|\cdot\rl, 
\label{eq:1clustbound:MeyerFC}
\end{alignat}
where we use the fact that $\topDist^2 = \topDist$ to simplify the last term in \eqref{eq:1clustbound:MeyerFC}.

While $|C^* \setminus \lcore| \leq \frac{|C^*|}{2}$, we will require a tighter bound on
$\sum_{i=1}^k |C^*_i \setminus \core_\ell(C^*_i)|$. Observe that, for any $j \notin
\cup_{i=1}^k \core_\ell(C^*_i)$, $d(j, S^*) > t_\ell \geq t^*_\ell$. So, by the definition
of $t^*_\ell$, there can be at most $\ell$ such  agents in $\C$. Hence, by summing
\eqref{eq:1clustbound:MeyerFC} over all clusters  $C^*_1, \ldots, C^*_k$, we obtain  
\begin{equation*}
\begin{split}
\Exp\Bigl[\sum_{j \in \mathcal{C}} \min(\delta_j, f)\, |\, \mathcal{E}\Bigr]
&\le (9 + 2\topDist)kf + (14 + 13\topDist)\sum_{j \in \C} (d(j, S^*) - t_\ell)^+ 
+ (3 + \topDist)\ell \cdot t_\ell \\
&\leq (9 + 2\topDist)kf + (14 + 13\topDist)\opt + (3 + \topDist)B.
\end{split} 
\end{equation*}
 		 
	The above bound is independent of the conditioning on $\mathcal{E}$, which can therefore be removed. Moreover, the upper bound on $\sum_{j \in \mathcal{C}} \mathbb{E}[\min(\delta_j, f)] $ can be used to establish an upper bound on the expected  cost induced by our solution $S$, as well as the expected size of $S$. Recall that $\topl(d(\mathcal{C}, S)) \le  \ell \cdot 4t_\ell + \sum_{j \in \mathcal{C}} (d(j, S) - (3 + \topDist)t_\ell)^+$, and $f = \frac{B}{k}$. We have
\begin{equation*}
\begin{split}
\E{\topl(d(\mathcal{C}, S))}
   &\le  \Exp\Bigl[ \ell \cdot (3 + \topDist)t_\ell + \sum_{j \in \mathcal{C}} (d(j, S) - (3 + \topDist)t_\ell)^+ \Bigr]\\
			&\le  (3 + \topDist)B + \sum_{i=1}^k \sum_{j \in C^*_i} \mathbb{E}[\min(\delta_j, f)]
			\leq (6 + 2\topDist)B + (9 + 2\topDist)kf + (14 + 13\topDist)\opt\\
   &\le  (15 + 4\topDist)B + (14 + 13\topDist)\opt. 
\end{split}
\end{equation*}
		
		\noindent We can also derive the following bound on the expected size of $S$:
\begin{equation*}
\begin{split}
	\sum_{i=1}^k \E{|S \cap C^*_i|} &\leq \sum_{i=1}^k \sum_{p \in C^*_i}  \frac{\E{\min(\delta_p, f)}}{f}\\ 
 &\leq \frac{(9 + 2\topDist)kf + (14 + 13\topDist)\opt + (3 + \topDist)B}{f}
\le  (26 + 16\topDist)k. 
\end{split}
 \end{equation*}

\noindent This completes the proof of Theorem~\ref{meyerson-topl-thm}. \qed

\subsection{Adaptive sampling for \boldmath $\ell$-centrum: Proof of
  Theorem~\ref{adsample-topl-thm}} \label{adsample-topl-proof} 


Fix an optimal solution $S^*=\{c^*_1,\ldots,c^*_k\}\in\F^k$.
Note that we are considering the case $\F=\C$ here.
Let $C^*_1,\ldots,C^*_k$ denote the clusters induced by $S^*$;
that is, for $q\in[k]$, $C^*_q\sse\C$ is the set of agents $j$ assigned
to center $c^*_q$ (i.e., $c^*_q=\topalt_{S^*}(j)$). 

The proof is a bit long, and somewhat technical, so we first give an outline. 
We consider the proxy cost 
$\sum_{j\in\C}(d(j,S)-\betaval t_\ell)^+$ (where $S$ is the center-set) as discussed
earlier; if this proxy cost is $O\bigl(\sum_{j\in\C}(d(j,S^*)-t_\ell)^+\bigr)$, then
since $t_\ell$ is a good estimate of $t^*_\ell$, we can easily infer that $S$ has
$\topl$-cost $O(\OPT)$.
  
The key property that we will show, which will be the technical crux of the proof, is
that if the $\topl$-cost of our solution is large, 
then the next center added to our solution $S$ lies in the ``core'' of some ``bad''
cluster, with some constant probability $p$ (Lemma~\ref{choosecore}).   
We define the notions of ``bad'' cluster and ``core'' of a cluster shortly,%
\footnote{The notion of core used here is similar to, but subtly different than, the one 
used in the analysis of Meyerson's algorithm in Section~\ref{append-meyerson-topl}.}
but, 
roughly speaking: 
(1) a bad cluster is a cluster $C^*_q$ whose points incur a large proxy cost compared
to $S^*$ (Definition~\ref{clgood}); (2) the core of a cluster $C^*_q$ consists of
points that are sufficiently close to its center $c^*_q$ (Definition~\ref{clcore}).
The idea here is that if every cluster is ``good'' (i.e., not bad), then the proxy cost will
be small and we will have bounded $\topl$-cost (Claim~\ref{claim:lgood}), and we will
argue that if $S$ contains a point from the core of a cluster, then that cluster is good
(Claim~\ref{coreisgood}).%

The upshot is that given the above property, in every iteration, we
make progress towards obtaining a low-cost solution 
by reducing the number of bad clusters with probability $p$. The expected number of bad
clusters thus decreases with each iteration, and we can then argue using standard
martingale arguments 
that after $(k+\sqrt{k})/p$ iterations, with some 
constant probability, we obtain a solution with no bad clusters.



\medskip
We now proceed with the details.
Let 
$\tau=28$, $\rho=35$. It will be convenient to analyze things in terms of the following
constants $\betaval=2$, $\alpha=3$, $\gm=4$, and $\kappa=8$; they are chosen to satisfy the
following inequalities: 
\begin{equation}
\begin{split}
\betaval & \geq 2,\ \ \gm = \al+1\geq\betaval, \ \ \al>1, \ \
1-\tfrac{\gm}{\rho}\geq 2\cdot\tfrac{\kp+\betaval}{\rho} \\
\kp & \geq\al+\betaval+3, \quad
\Bigl(1-\tfrac{\gm}{\rho}\Bigr)\cdot\tfrac{\al-1}{2\al\kp}\geq\frac{1}{\tau}.
\end{split} 
\label{abkchoices}
\end{equation}



						
\begin{definition} \label{clgood}
Say that a cluster $C^*_q$ is $\ell$-\emph{good}, if 
$\sum_{j \in C^*_q}(d(j, S)-\betaval t_\ell)^+ \leq \gamma \sum_{j \in C^*_q}(d(j, c^*_q) - t_\ell)^+$. 
If $C^*_q$ is not $\ell$-good, it is {\em $\ell$-bad}.
\end{definition}
						
\begin{claim}\label{claim:lgood} \label{allgood}
If every cluster is $\ell$-good, then 
$\topl(d(\mathcal{C}, S)) \leq(1+\varepsilon)\gamma\cdot\topl(d(\mathcal{C}, S^*)).$ 
\end{claim}

\begin{proof}
By Claim~\ref{csproxy} (b) (and since $\C=\bigcup_{q=1}^kC^*_q$), we have 
\begin{equation*}
\begin{split}
\topl(d(\C,S)) & \leq\ell\cdot\betaval t_\ell+
\sum_{q=1}^k\sum_{j\in C^*_q}(d(j,S)-\betaval t_\ell)^+ \\
& \leq \betaval\ell\max\Bigl\{(1+\ve)t^*_\ell,\,\ve\cdot\tfrac{\OPT}{\ell}\Bigr\}+
\sum_{q=1}^k\gm\cdot\sum_{j\in C^*_q}(d(j,c^*_q)-t^*_\ell)^+ \\
& \leq \gm\cdot\max\bigl\{(1+\ve)\ell t^*_\ell,\,\ve\OPT\bigr\}+
\gm\cdot\sum_{j\in\C}(d(j,S^*)-t^*_\ell)^+\leq \gm(1+\ve)\OPT.
\end{split}
\end{equation*}
The second inequality follows since 
$t^*_\ell\leq t_\ell\leq\max\bigl\{(1+\ve)t^*_\ell,\ve\cdot\frac{\OPT}{\ell}\bigr\}$,
and since all clusters are $\ell$-good, and the third is because $\gm\geq\betaval$. 
The bound in the claim follows. 
\end{proof}

We now define the core of a $C^*_q$ cluster to consist of points in
$C^*_q$ that are close to $c^*_q$, where the definition of close is tailored to ensure
that if a center lies in the core of $C^*_q$, then $C^*_q$ is $\ell$-good.
Define the {\em radius} of $C^*_q$ to be 
$r_\ell(C^*_q) = \frac{\sum_{j \in C^*_q} (d(j, c^*_q)-t_\ell)^+}{|C^*_q|}$. 
%
%
%
For the precise definition of core, we proceed somewhat differently from Aggarwal et. al,
due to the nature of the proxy cost that we are working with, which does not satisfy the
triangle inequality.  
In particular, we need to define things differently depending on whether the center $c^*_q$
is close or far away from the current center set.


\begin{definition} \label{clclose}
We say that a cluster $C^*_q$ (with center $c^*_q$) is 
{\em $\ell$-close} 
$d(c^*_q,S) \leq \kappa \cdot \max \{t_\ell, r_\ell(C^*_q)\}$;  
otherwise, $C^*_q$ is {\em $\ell$-far}. 
\end{definition} 
						
						
\begin{definition} \label{clcore}
The $\ell$-core, $\core_\ell(C^*_q)$, of a cluster $C^*_q$ is defined as: 
\begin{equation*}
\begin{cases}
\{j \in C^*_q : d(j, c^*_q) \leq t_\ell\}; & \text{if $C^*_q$ is $\ell$-close} \\ 
\{j \in C^*_q : (d(j, c^*_q) - t_\ell)^+ \leq \alpha \cdot r_\ell(C^*_q)\};
& \text{otherwise}.
\end{cases}
\end{equation*}
\end{definition}

In the sequel, we will simply say core to refer to the $\ell$-core.
We note that the notions of $\ell$-\{good, bad, close, far\}, and hence, also the notion of
core, are all relative to the current center set. Clearly, since the center-set only
expands, once a cluster becomes $\ell$-good or $\ell$-close, it retains that property
throughout. 

\begin{claim} \label{coreisgood}
Consider a cluster $C^*_q$, and let $S$ be the current center-set. 
If $S\cap\core(C^*_q)\neq\es$, then $C^*_q$ is $\ell$-good (and hence remains $\ell$-good
throughout). 
\end{claim}

\begin{proof}
Let $s$ be a point in $S\cap\core(C^*_q)$.
We have
\begin{equation*}
\begin{split}
\sum_{j \in C^*_q}(d(j,s) - \betaval t_\ell)^+ 
& \leq \sum_{j\in C^*_q}\bigl((d(j,c^*_q)-t_\ell)^++(d(s,c^*_q)-t_\ell)^+\bigr) \\
& = \sum_{j \in C^*_q} (d(j, c^*_q) - t_\ell)^++|C^*_q|(d(s, c^*_q) - t_\ell)^+. 
\end{split}
\end{equation*}
The inequality follows from the triangle inequality applied to $d$, and since
$(y+z)^+\leq y^++z^+$.
Since $s\in\core(C^*_q)$, the
second term in the final inequality above is at most 
$\alpha |C^*_q|r_\ell(C^*_q)$; note that this holds both when $C^*_q$ is $\ell$-close and
is $\ell$-far. 
So 
we have 
$\sum_{j \in C^*_q}(d(j,s) - \betaval t_\ell)^+\leq 
(1 + \alpha)\sum_{j \in C^*_q} (d(j, c^*_q) - t_\ell)^+$,
showing that $C^*_q$ is $\ell$-good.
\end{proof}


Lemma~\ref{choosecore} is the key property that we show. We defer its proof, which is
rather technical, and first show that given this, adaptive sampling returns a constant-factor
solution with constant probability.
						
\begin{lemma}\label{lem:chooseFromCore} \label{choosecore}
Consider any iteration $i$, and suppose that
$\topl(d(\C, S_{i-1}))>\rho {(1 + \ve)} \cdot \topl(d(\C,S^*))$. 
Then 
$\Pr[\text{$s_i$ lies in the core of an $\ell$-bad cluster}]\geq\frac{1}{\tau}$.
\end{lemma}


\begin{proof}[Finishing up the proof of Theorem~\ref{adsample-topl-thm}]
Given Lemma~\ref{choosecore}, the proof proceeds via a standard martingale property along
the lines of that used by~\cite{AggarwalDK09}. Let $p=1/\tau$ and 
let $N=\ceil{\tau(k+\sqrt{k})}$. 
Recall that $S_i$ is the center-set at the {\em start} of iteration $i+1$ (and end of
iteration $i$), for $i\geq 0$.
Intuitively, we would like to define $X_i$ as the number of bad clusters at the end of
iteration $i$ (with $X_0=k$), and consider a shifted version of this to obtain a
supermartingale, but $X_i-X_{i+1}$ could potentially be large, so we need to proceed a bit
more carefully.
Define $X_0=k$. For $i\geq 1$, define $X_i=X_{i-1}-1$ if the core of some bad cluster was hit
in iteration $i$, or $\topl(d(\C,S_{i-1}))\leq\rho(1+\ve)\OPT$, and set
$X_i=X_{i-1}$ otherwise.
Formally, if $s_i\cap\core(C^*_q)\neq\es$ for some bad cluster $C^*_q$ with respect
to center-set $S_{i-1}$, or $\topl(d(\C,S_{i-1}))\leq\rho(1+\ve)\OPT$, then
$X_i=X_{i-1}-1$; otherwise $X_i=X_{i-1}$. 
Note that we have $\E{X_{i}|X_{i-1}}\leq X_{i-1}-p$: if $\topl(d(\C,S_{i-1}))>\rho(1+\ve)\OPT$,
this follows due to Lemma~\ref{choosecore}.

Observe that if $X_N=0$, then $\topl(d(\C,S_N))\leq\rho(1+\ve)\OPT$: either we have
$\topl(d(\C,S_{N-1}))\leq\rho(1+\ve)\OPT$; if not, then by Claim~\ref{coreisgood}, the
number of bad clusters at the end of iteration $N$ is at most $X_N=0$, and hence by
Claim~\ref{claim:lgood}, we have $\topl(d(\C,S_N))\leq\rho(1+\ve)\OPT$.
So if we show that $\Pr[X_N>0]\leq e^{-p/4}$, 
then we are done.
For $i=0,1,\ldots$, define $Y_i=X_i+i\cdot p$. Then, we have $|Y_{i+1}-Y_i|\leq 1$
for all $i\geq 0$, and $\E{Y_{i+1}|Y_0,\ldots,Y_i}\leq X_{i+1}-p+(i+1)\cdot p=Y_i$, so
$Y_0,Y_1,\ldots$ form a super-martingale.
Now if $X_N>0$, we have $Y_N>Np$. By the Azuma-Hoeffding inequality, we have
\begin{equation*}
\Pr[Y_N-Y_0>(Np-k)] \leq\exp\Bigl(-\tfrac{(Np-k)^2}{2N}\Bigr) 
\leq\exp\Bigl(-\tfrac{kp}{2(k+\sqrt{k})}\Bigr)\leq e^{-\frac{p}{4}}. \qedhere
\end{equation*}
\end{proof}

\begin{proof}[Proof of Lemma \ref{lem:chooseFromCore}]
Let $Z^*\in\{C^*_1,\ldots,C^*_q\}$ be the random cluster containing the sampled point
$s_i$. 
Throughout, we use $S$ to denote $S_{i-1}$, the center-set at the start of
iteration $i$.
For convenience, define the following index-sets, where $\ell$-\{good, bad, close, far\}
are all with respect to $S$. 
\begin{enumerate}[label=$\bullet$, topsep=2pt, noitemsep, leftmargin=*]
\item $\good=\{q\in[k]: C^*_q\text{ is $\ell$-good}\}$, \qquad
$\bad=\{q\in[k]: C^*_q\text{ is $\ell$-bad}\}$
\item $\close=\{q\in[k]: C^*_q\text{ is $\ell$-close}\}$, \qquad
$\far=\{q\in[k]: C^*_q\text{ is $\ell$-far}\}$.
\end{enumerate}
We first show that with constant probability, $Z^*$ is an $\ell$-bad
cluster (Lemma~\ref{badprob}). 
Then, we show that conditioned on $Z^*$ being an $\ell$-bad, $\ell$-far
cluster, we have that $s_i\in\core(Z^*)$ with constant probability (Lemma~\ref{corebadfar}).
Next, we show that the probability that $Z^*$ is $\ell$-close and $s_i\notin\core(Z^*)$ is
small (Lemma~\ref{coreclose}). Finally, we put these together to finish up the proof.
  
\begin{lemma}\label{claim:inBadClust} \label{badprob}
$\Pr[Z^*\text{ is $\ell$-bad}] \geq  1-\frac{\gamma}{\rho}$.
\end{lemma}

\begin{proof} 
The probability that $Z^*$ is $\ell$-good is 
$\frac{\sum_{q\in\good}\sum_{j\in C^*_q}(d(j,S)-\betaval t_\ell)^+}{\sum_{j\in\C}(d(j,S)-\betaval t_\ell)^+}$,
which is at most
\[
\frac{\betaval t_{\ell } \cdot \ell+\sum_{q\in\good}\sum_{j\in C^*_q}(d(j, c^*_q)-\betaval
  t_\ell)^+}
{\betaval t_{\ell} \cdot \ell + \sum_{j \in \mathcal{C}} (d(j, S) - \betaval t_{\ell})^+}.
\]
The denominator above is at least $\topl(d(\C,S))$, by Claim~\ref{csproxy} (b), and so at
least $\rho(1+\ve)\OPT$. We upper bound the numerator.
By the definition of $\ell$-good clusters and since $t_\ell\geq t^*_\ell$,  
the second term in the numerator is at most 
$\sum_{q\in\good}\gm\sum_{j\in C^*_q}(d(j,c^*_q)-t^*_\ell)^+$.
So the above expression is at most 
\begin{gather*}
\frac{\betaval\max\{(1+\ve)\ell t^*_\ell,\ve\OPT\}+\gm\sum_{j\in\C}(d(j,S^*)-t^*_\ell)^+}{\rho(1+\ve)\OPT}
\end{gather*}
which is at most $\frac{\gm(1+\ve)}{\rho(1+\ve)}$, 
where we use that $\gm\geq\betaval$. 
\end{proof} 
							
We next consider the cases $Z^*$ is $\ell$-far and $Z^*$ is $\ell$-close separately. 
Conditioned on $Z^*$ being 
$\ell$-far, we show that $s_i\in\core(Z^*)$ with constant probability. 

\begin{lemma}\label{claim:Case2Far} \label{corebadfar}
Consider any 
$\ell$-far cluster $C^*_q$. Then
$\Pr[s_i\in\core(Z^*)\,|\,Z^*=C^*_q]\geq\frac{\al-1}{\al\kp}$.
\end{lemma}

\begin{proof} 
The probability is $\Pr[s_i\in\core_\ell(C^*_q)]/\Pr[s_i\in C^*_q]$.
We abbreviate $r_\ell(C^*_q)$ to $r_\ell$ in this proof, since we are considering the
fixed cluster $C^*_q$.
Since $C^*_q$ is $\ell$-far, 
$\core_\ell(C^*_q) = \{j \in C^*_q: (d(j,c^*_q) - t_\ell)^+ \leq \alpha \cdot r_\ell\}$. 
As $|C^*_q|\cdot r_\ell$ is at least 
$\sum_{j \notin\core_{\ell}(C^*_q)}(d(j,c^*_q) -  t_{\ell})^+
\geq |C^*_q\setminus\core_\ell(C^*_q)|\cdot\alpha r_{\ell}$, 
we have $|\core_\ell(C^*_q)| \geq \frac{\alpha - 1}{\alpha} \cdot |C^*_q|$. 
%
We have  
\begin{equation*}
\begin{split}
\frac{\Pr[s_i\in\core_{\ell}(C^*_q)]}{\Pr[s_i\in C^*_q]}
& = \frac{\sum_{j \in\core_{\ell}(C^*_q)}(d(j,S) - \betaval t_{\ell})^+}
{\sum_{j \in C^*_q} (d(j,S) - \betaval t_\ell)^+} 
\geq\frac{\sum_{j \in\core_\ell(C^*_q)}(d(c^*_q,S) - d(j,c^*_q) - \betaval t_\ell)^+}
{\sum_{j \in C^*_q} (d(j,c^*_q) + d(c^*_q, S) - \betaval t_\ell)^+} \\
& \geq  \frac{|\core_\ell(C^*_q)| \cdot (d(c^*_q, S) -\alpha  r_{\ell} -(\betaval+1)t_{\ell})}
{|C^*_q| \cdot (r_{\ell}+(d(c^*_q,S)-(\betaval-1)t_{\ell}))}
\geq \frac{\alpha - 1}{\alpha} \cdot 
\frac{d(c^*_q, S) -\alpha  r_{\ell} - (\betaval+1)t_{\ell}}{r_{\ell} + d(c^*_q,S)}
\end{split}
\end{equation*}
%
The second inequality is because $d(j, c^*_q) - t_\ell \leq\al r_\ell$ for all 
$j \in\core_\ell(C^*)$, 
and because $d(c^*_q,S)\geq\kp\max\{t_\ell,r_\ell\}\geq(\betaval-1)t_\ell$,
as $\kp\geq\al+\betaval+1$.
The final expression above is an increasing function of $d(c^*_q,S)$, and so since $C^*_q$
is $\ell$-far, we have 
\begin{equation*}
\frac{d(c^*_q,S) -\alpha  r_{\ell} - (\betaval+1)t_{\ell}}{r_{\ell} + d(c^*_q,S)}
\geq\frac{(\kp-\al-\betaval-1)\max\{r_\ell,t_\ell\}}{(\kp+1)\max\{r_\ell,t_\ell\}}
\geq\frac{2}{\kp+1}\geq\frac{1}{\kp}. \quad \text{(due to \eqref{abkchoices})}
\qedhere
\end{equation*}
\end{proof} 	

Next, we consider the case where $Z^*$ is $\ell$-close.

\begin{lemma}\label{claim:Case1Close} \label{coreclose}
$\Pr[Z^*\text{ is $\ell$-close}, s_i\notin\core(Z^*)]\leq\frac{\kp+\betaval}{\rho}$.
\end{lemma}

\begin{proof} 
The given probability is 
\begin{equation*}
\frac{\sum_{q\in\close} \sum_{j \in C^*_q\sm\core_{\ell}(C^*_q)} (d(j, S) - \betaval
  t_{\ell}) ^+ }{\sum_{j \in \mathcal{C}} (d(j, S) - \betaval t_\ell)^+ }
\leq \frac{\ell\cdot\betaval t_\ell+\sum_{q\in\close} \sum_{j \in C^*_q\sm\core_{\ell}(C^*_q)} (d(j, S) - \betaval
  t_{\ell}) ^+ }{\ell\cdot\betaval t_\ell+\sum_{j \in \mathcal{C}} (d(j, S) - \betaval t_\ell)^+ }
\end{equation*}

\vspace*{-2ex}
\begin{alignat*}{1}
\quad 
& \leq \frac{\ell\cdot\betaval t_\ell+\sum_{q\in\close}\sum_{j \in C^*_q\sm\core_{\ell}(C^*_q)} 
\bigl(d(j,c^*_q)+d(c^*_q,S)-\betaval t_{\ell}\bigr)^+ }{\rho(1+\ve)\OPT} \\
& \leq \frac{\ell\cdot\betaval
  t_\ell+\sum_{j\in\C}\bigl(d(j,S^*)-t_{\ell}\bigr)^+}{\rho(1+\ve)\OPT}
+\frac{\sum_{q\in\close}|C^*_q\sm\core(C^*_q)|(d(c^*_q,S)-t_\ell)^+}{\rho(1+\ve)\OPT}.
\end{alignat*}
The first term above can be bounded by $\frac{\betaval}{\rho}$, using the bounds on
$t_\ell$, by arguing as in the proof of Lemma~\ref{badprob}. 
To bound the second term, we observe that every point 
$j\in\bigcup_{q\in\close}\bigl(C^*_q\sm\core(C^*_q)\bigr)$ has 
$d(j,S^*)\geq t_\ell\geq t^*_\ell$. So by definition of $t^*_\ell$, there can be at most
$\ell$ such points in total. Also, for each $q\in\close$, we have
$d(c^*_q,S)\leq\kp\max\{t_\ell,r_\ell(C^*_q)\}$. 
Therefore, 
\begin{equation*}
\begin{split}
\frac{\sum_{q\in\close}|C^*_q\sm\core(C^*_q)|(d(c^*_q,S)-t_\ell)^+}{\rho(1+\ve)\OPT} 
& \leq 
\frac{\kp\cdot\sum_{q\in\close}|C^*_q\sm\core(C^*_q)|(t_\ell+r_\ell(C^*_q))}{\rho(1+\ve)\OPT} \\
& \leq 
\frac{\kp\cdot\bigl(\ell\cdot
  t_\ell+\sum_{q\in\close}|C^*_q|r_\ell(C^*_q)\bigr)}{\rho(1+\ve)\OPT}
\leq\frac{\kp(1+\ve)}{\rho(1+\ve)}.
\end{split}
\end{equation*}
Putting these bounds together, we obtain that
$\Pr[Z^*\text{ is $\ell$-close}, s_i\notin\core(Z^*)]\leq\frac{\kp+\betaval}{\rho}$.
\end{proof}
								
\medskip\noindent 
Finally, we combine Lemmas~\ref{badprob}--\ref{coreclose} to lower bound
$\Pr[Z^*\text{ is $\ell$-bad},\,s_i\in\core_\ell(Z^*)]$. 
This probability is 
\begin{equation*}
\begin{split}
\Pr[Z^*&\text{ is $\ell$-bad, $\ell$-far}]\cdot
\Pr[s_i\in\core_\ell(Z^*)\,|\,Z^*\text{ is $\ell$-bad, $\ell$-far}]
+ \Pr[Z^*\text{ is $\ell$-bad, $\ell$-close}] \\
& -\Pr[Z^*\text{ is $\ell$-bad, $\ell$-close, $s_i\notin\core_\ell(Z^*)$}]
\end{split}
\end{equation*}
Define $\tht_\far=\Pr[Z^*\text{ is $\ell$-bad, $\ell$-far}]$. Similarly, let
$\tht_\close=\Pr[Z^*\text{ is $\ell$-bad, $\ell$-close}]$.
Then, $\Pr[Z^*\text{ is $\ell$-bad}, s_i\in\core_\ell(Z^*)]$ is at least
\begin{alignat}{1}
\sum_{q\in\bad\cap\far}&\Pr[Z^*=C^*_q]\cdot
\Pr[s_i\in\core_\ell(C^*_q)\,|\,Z^*=C^*_q] 
+ \max\Bigl\{0,\tht_\close-\Pr[Z^*\text{ is $\ell$-close,
    $s_i\notin\core_\ell(Z^*)$}]\Bigr\} 
\notag \\
& \geq \tht_\far\cdot\frac{\al-1}{\al\kp}+
\max\Bigl\{0,\tht_\close-\tfrac{\kp+\betaval}{\rho}\Bigr\} \label{prexpr}
\end{alignat}
where the last inequality follows from Lemmas~\ref{corebadfar} and~\ref{coreclose}.
Notice that $\tht_\far+\tht_\close=\Pr[Z^*\text{ is $\ell$-bad}]\geq
1-\frac{\gm}{\rho}\geq 2\cdot\frac{\kp+\betaval}{\rho}$ 
(by Lemma~\ref{badprob} and \eqref{abkchoices}).
If $\tht_\far\geq\frac{1}{2}\bigl(1-\frac{\gm}{\rho}\bigr)$, then \eqref{prexpr} is at least
$\bigl(1-\frac{\gm}{\rho}\bigr)\cdot\frac{\al-1}{2\al\kp}$.
Otherwise, we have \eqref{prexpr} is at least 
$\frac{1}{2}\bigl(1-\frac{\gm}{\rho}\bigr)-\tht_\far\bigl(1-\frac{\al-1}{\al\kp}\bigr)
\geq\bigl(1-\frac{\gm}{\rho}\bigr)\cdot\frac{\al-1}{2\al\kp}$.
So the desired probability is at least 
$\bigl(1-\frac{\gm}{\rho}\bigr)\cdot\frac{\al-1}{2\al\kp}\geq\frac{1}{\tau}$.

This completes the proof of Lemma~\ref{choosecore}, and hence
Theorem~\ref{adsample-topl-thm}. 
\end{proof}

\section{Extension to the setting \boldmath $\F\neq\C$}
\label{sec:A_neq_C} \label{extn}
We now consider the more general setting where $\F\neq\C$.
While with cardinal information, it is easy enough to reduce this to the earlier case
(for instance, by moving agents to the alternatives nearest to them), various challenges
arise when we seek to limit the number of value queries because, we cannot query an
alternative $a\in\F$ for distances to agents. 
With suitable, relatively minor, changes, our mechanisms with per-agent query complexity
bounds can be extended to this more general setting. 

Recall that our mechanisms in Sections~\ref{blackbox} and~\ref{sec:adsample} comprise
two main ingredients, obtaining an estimate of $\OPT$, and leveraging this estimate.
We need to make changes to both ingredients.
We need to modify how we compute the estimates on $\OPT$ using
Mechanisms~\ref{boruvka} and~\ref{kcenter}. 
Second, we need to make slight changes
to the $\ell$-centrum extensions of Meyerson's algorithm 
and adaptive sampling (i.e., Algorithms~\ref{meyerson-topl} and~\ref{adsample-topl}). 
The latter change, in both algorithms, is of a similar form, where we still use an agent
$s\in\C$%
---either the ``newly arrived'' agent in Meyerson's algorithm,
or an agent that is sampled in adaptive sampling---to base our decision, but we add the
alternative $\topalt(s)$ to our center-set; see Algorithm~\ref{meyerson-toplgen}
and Algorithm~\ref{adsample-toplgen} in Section~\ref{mechanisms-gen}. 
 
\subsection{Computing estimates of $\opt$} \label{optestim-gen}

\paragraph{Modified Boruvka mechanism.}
In the setting where $\F = \C$, we had leveraged the fact
that the cost of a minimum-cost $k$-forest in the complete graph on $\C$ estimates
$\opt$ within a factor of $n$. 
However, when $\F\neq \C$, the minimum-cost $k$-forest (in the complete bipartite graph on
$\F \cup \C$), $F^*$, may include unnecessary candidate-voter edges, or singleton
voter-components, and
therefore, the cost of $F^*$ need be bounded with respect to $\opt$. To circumvent the
first issue, we will only consider candidates in $\tdF := \{\topalt(j) : j \in \C\}$. If
$F^*$ is a minimum-cost $k$-forest in the complete bipartite graph on $\tdF \cup \C$,
we show that the cost of the subgraph $H = F^* \cup \{(j,\topalt(j)) : j \in \C\}$ can again
be used to obtain an $O(n)$-approximate estimate of $\opt$.  

\begin{claim}\label{claim:estimate-FC}
    Let $F^*$ be a minimum-cost $k$-forest in the complete bipartite graph on $\F \cup \C$, and define $H = F^* \cup \{(j,\topalt(j)) : j \in \C\}$. 
Then $d(H)\leq 5\opt_n \leq 5n \cdot d(H)$,
where $d(H) = \sum_{e \in H} d_e$. 
\end{claim}
\begin{proof} 
We abbreviate $\opt_n$ to $\opt$.
Let $\tG$ be the complete bipartite graph on $\tdF \cup \C$, and let $S^* \subseteq
\F$ be an optimal $k$-median solution.  It is possible that $S^*$ contains $i \in \F
\setminus \tdF$; in this case, we cannot directly use $S^*$ to construct a $k$-forest
in $\tG$. Instead,  we will use $S^*$ to construct a new solution $\tS
\subseteq \tdF$ of cost no more than $3 \opt$ (and then use $\tS$ to construct a
$k$-forest in $\tG$). For each  $i \in S^*$, define $\phi(i) = \arg\min_{j \in \C}
d(i, j) + d(j,\topalt(j))$. By the triangle inequality, the distance between $j \in \C$ and
$\phi(i)$ is at most  $d(i, j) + d(i, \phi(i)) + d(\phi(i), \topalt(\phi(i))) \leq 2 d(i, j) +
d(j,\topalt(j))$. Hence, if $\tS = \{\topalt(\phi(i)) : i \in S^*\}$, we have
\[  \sum_{j \in \C} \min_{i \in \tS} d(i, j) \leq  \sum_{j \in \C} \left(2\min_{i \in S^*} d(i, j) +  d(j,\topalt(j))  \right) \leq 3 \opt \] 
Given $\tS$, define $\tilde x(j) = \arg\min_{i \in \tilde S} d(i, j)$. Let $F = \{
(j, \tilde x(j)) : j \in \C \} \cup \{ (j,\topalt(i)) : i \in \tilde A \setminus \tilde
S\}$. Observe that $F$ is a $k$-forest in $\tG$, so $d(F) \geq d(F^*)$. Moreover,
the cost of $F$ is at most $3 \opt  + \sum_{j\in\C} d(j,\topalt(j)) \leq 4 \opt $.  
It immediately follows that 
$d(H) \leq 4 \opt  + \sum_{j \in \C}d(j,\topalt(j)) \leq 5 \opt$.

We now prove the upper bound on $\opt$. For each component $C$ induced by $H$, choose an
arbitrary cluster center in $C \cap \tdF$; let $S$ be the set of these centers. Since
$H$ has at most $k$ components and all components of $H$ have a size of at least 2,  this
is a well-defined operation, and does indeed yield a feasible $k$-median solution. For any
$j \in \C$ and $i \in \tdF$ that lie in the same component of $H$, we can bound $d(i,j)$ by
the cost of this component;
so summing over all clients, we obtain that $\opt \leq n \cdot d(H)$.  
\end{proof}


We can compute $d(H)$ as defined in Claim \ref{claim:estimate-FC} using a modification of
Boruvka's algorithm, which again requires only $O(\log n)$ value queries per agent.  
\begin{namedmech}[h!]{\boruvkagen}
	\caption{\hfill Modification of Boruvka's algorithm for min-cost $k$-forest}\label{alg:Boruvka:FC}
\textbf{Input}: {Preference profile $\sigma$.}   
	\begin{algorithmic}[1]
	\STATE Fix a tie-breaking rule on the edges (that will be used in all subsequent edge-cost comparisons). \\
	\STATE $F \leftarrow \varnothing$,  $V_1 \leftarrow \tilde A \cup  \mathcal{C}$,  $E_1 \leftarrow \{ \{i, j \} : i \in \tilde A, \,  j \in \mathcal{C} \}$,  $t \leftarrow 1$ \\
	\WHILE{$|V_t| > 1$}{ 
	\FOR{$S \in V_t$}{
	\STATE For each $v \in S$, query the value of $\min_{e \in \delta(v) \cap \delta(S)}d(e)$\\
	\STATE Add $e = \arg\min_{e' \in \delta(S)} d(e')$ to $F$
	}\ENDFOR
	\STATE Contract the components of $G_t = (V_t, F \cap E_t)$ into supernodes to get the (multi)graph $G_{t+1} = (V_{t+1}, E_{t+1})$\\
	$t \leftarrow t + 1$
	}\ENDWHILE  
    \STATE Remove the $k-1$ heaviest edges in $F$.\\
	\RETURN{$ n \cdot \left(\sum_{e \in F} d(e) + \sum_{j \in \C} d(j,\topalt(j))\right)$}.
	\end{algorithmic}
\end{namedmech}
\FloatBarrier

\paragraph{Modified $k$-center mechanism.}
We can also modify Mechanism~\ref{kcenter} (in Section~\ref{kcentkmed}) to the setting
$\F\neq\C$ as described below.

\begin{theorem} 
In the setting where $\F \neq \C$, if we modify Mechanism \ref{alg:kcenter} to open a
center at $\topalt(s_t)$ in each iteration $t$, the resulting solution has cost at most $3
\cdot \opt_1$.  
\end{theorem}
\begin{proof}
    Let $S$ be the set of centers opened by Mechanism \ref{alg:kcenter:FC}, and let $C^*_1, \ldots, C^*_k$ be the clusters induced by an optimal solution $S^*$, with centers $c^*_1, \ldots, c^*_k$ respectively. Notice that, for any $j_1, j_2 \in C^*_i$, $d(j_1, j_2) \leq d(j_1, c^*_i) + d(j_2, c^*_i) \leq 2 \opt_1$, by the triangle inequality. 

    If $S$ opens exactly one center in each cluster $C^*_i$, then by the earlier observation, the distance between any agent $j \in \C$ and the closest center in $S$ is at most $2\opt_1$. 
    
    Otherwise, some cluster $C^*_i$ contains two centers opened by $S$. This is only
    possible if,  at some step $t$ after $a_1 \in C^*_i$ was opened,  the agent $s_t$
    selected in that step has $\topalt(s_t) = a_2 \in C^*_i$. 
By the triangle inequality, the distance from $s_t$ to $a_1$ is at
    most $d(a_1, a_2) + d(a_2, s_t) \leq 3 \opt_1$. 
By construction, $s_t$ is farthest from the currently open center-set $S_{t-1}$,   
so we have $d(j,S_{t-1})\leq 3\opt_1$ for every $j \in \C$. 
Thus, in both cases, we have $\max_{j \in \C} d(j, S) \leq 3\opt_1$.  
\end{proof}
 
\FloatBarrier
\subsection{Constant-factor distortion mechanisms} \label{mechanisms-gen}

\paragraph{\boldmath Meyerson coupled with black-box reduction.}
Algorithm~\ref{meyerson-toplgen} in Section~\ref{mey-toplproof} adapts
Algorithm~\ref{meyerson-topl}   
to the setting $\F\neq\C$, and Theorem~\ref{meytoplgen-thm} analyzes its performance
guarantee. 
As before, combining Mechanism~\ref{alg:Boruvka:FC}, which estimates $\opt$,
Algorithm~\ref{alg:Meyerson-topL:FC}, which is used to find a bicriteria solution to
sparsify the instance, and our 
black-box reduction, yields the following mechanism, which is
an adaptation of Mechanism~\ref{alg:mech1L} to the $\F \neq \C$ setting.  

    \begin{namedmech}[h!]{\bbmeyersongen} 
        \caption{
          \hfill $O(\log k \log n)$ per-agent query complexity when $\F \neq \C$}
        \label{alg:mech1L:FC} \label{bb-meyerson-gen}
        \textbf{Input}: {Preference profile $\sigma$, $\rho$-approximation $\mathcal{A}$
          for $\ell$-centrum, where $\rho=O(1)$}\\
        \vspace{-10pt}
        \begin{algorithmic}[1]
               \STATE $\mathcal{S} \leftarrow \{S_0\} $ where $S_0$ is an arbitrary set of $k$ centers\\
            \STATE $B'$: output of Mechanism \ref{alg:Boruvka:FC} \\ 
            \STATE $x_1, \ldots, x_n$: Randomly shuffled sequence of agents\\ 
            \FOR{$i=1,\ldots, \lceil \log_2 5n^2 \rceil + 1$} 
            \STATE $B_i \leftarrow 2^{i-1} \cdot B'/n^2$, 	$f \leftarrow B_i/k$ \\ 
            \REPEATN{$\log(1/\delta)$} \label{line:Mech1L-meyerSTART:FC}
            \STATE $S$: output of Algorithm \ref*{alg:Meyerson-topL:FC} with $B = B_i$. 
            \IF{$|S| \le  120k$}
            \STATE $\mathcal{S} \leftarrow \mathcal{S} \cup \{S\}$;
            compute $d(\C,S)$ using one query per agent
            \ENDIF 
            \ENDREPEAT  
            \ENDFOR  
	    \STATE If $\Sc=\es$, {\bf return} failure. 
            Otherwise, let $\bS \assign \arg\min\limits_{S \in\mathcal{S}}\topl(d(\mathcal{C}, S))$. 
            For $i\in\bS$, set $w_i=\bigl|\{j\in\C: \topalt_{\bS}(j)=i\}\bigr|$; 
            for all $i\notin\bS$, set $w_i=0$.
            \RETURN Mechanism~\bbtopl$(\bS,\sigma,w,B',\mathcal{A})$ 
        \end{algorithmic}
        
    \end{namedmech}

    \FloatBarrier

The same arguments that lead to the proof of Theorem~\ref{bb-meyerson-thm} yield the
following guarantee.

\begin{theorem} \label{thm:mech1:FneqC} \label{gen-bbmeyerson-thm}
Mechanism~\ref{alg:mech1L:FC} has $O(\log k \log n)$ per-agent query complexity, and
achieves $O(1)$-distortion for the $\ell$-centrum problem with probability at least $1-\dt$.
    \end{theorem}
	
\paragraph{Adaptive-sampling mechanism.}
The following slight change to Algorithm~\ref{adsample-topl} modifies it to work in the
$\F\neq\C$ setting.
	\begin{namedalg}[h!]{\adsampltoplgen}
		\caption{\hfill Adaptive sampling algorithm for $\ell$-centrum when $\F \neq
                  \C$} \label{alg:Lcentrum:FC} \label{adsample-toplgen} \label{adsample-aneqc-alg}
		\textbf{Input}: An $\ell$-centrum instance $(\mathcal{C}, \F, d)$, positive integer $\tau$, and  guess for $t^*_\ell$ ($t_\ell$)
		\begin{algorithmic}[1]
			\STATE $S_0 \leftarrow \varnothing$ \; 
			\FOR{$i=1,\ldots,\ceil{38(k + \sqrt{k})}$}
			\STATE Sample $s_i$ with probability proportional to $(d(s_i, S_{i-1}) - \gbeta t_\ell)^+$ \\ 
			\STATE Update $S_i \leftarrow S_{i-1} \cup \{\topalt(s_i)\}$. 
			\ENDFOR 
			\RETURN $S_{\ceil{38(k + \sqrt{k})}}$\\
		\end{algorithmic}
	\end{namedalg}

\begin{theorem} \label{adsample-aneqc-thm}
Let $t_\ell$ be such that
$t^*_\ell\leq t_\ell\leq\max\{(1+\ve)t^*_\ell,\frac{\ve\OPT}{\ell}\}$, for some $\ve>0$. 
Algorithm~\ref{adsample-aneqc-alg} run with parameter 
$t_\ell$ opens at most $76k$ centers,  
and returns a solution of $\topl$-cost at most $35(1+\ve)\cdot\OPT$ with constant
probability.    
\end{theorem} 

Algorithm~\ref{adsample-toplgen} leads to the following corresponding mechanism. 

\begin{namedmech}[h!]{\adsamplagentgen}
			\caption{\hfill $O(k\log \ell)$ per-agent query complexity when $\F \neq \C$}
			\label{alg:mech2L:FC} \label{adsample-gen}
			\textbf{Input}: {Preference profile $\sigma$, $\rho$-approximation
                          $\mathcal{A}$ for $\ell$-centrum, where $\rho=O(1)$}
			\begin{algorithmic}[1]
                \STATE $S_0$: Output of modified Mechanism \ref{alg:kcenter}, $B = \max_{j \in \C} d(j, S_0)$.  
				\STATE $\mathcal{T} = \{\ell \cdot B_1 (1 + \varepsilon)^{-r} : r = 0, \ldots, \log_{1 + \varepsilon}(\frac{3\ell}{\varepsilon})\}$,  $\mathcal{S} \leftarrow \varnothing$  
				\FOR{$t_\ell \in \mathcal{T}$}
				\REPEATN{$\log(1/\delta)$}  
				\STATE $S$: output of Algorithm \ref*{alg:Lcentrum:FC}
                                using parameter $t_\ell$ \\ 
				\STATE $\mathcal{S} \leftarrow \mathcal{S} \cup \{S\}$;
                                compute $d(\C,S)$ using one query per agent
				\ENDREPEAT 
				\ENDFOR
				
				\STATE Let $\bS\assign\arg\min\limits_{S \in\mathcal{S}}\topl(d(\C,S))$  
                                \STATE Query $d(j,a)$ for all $j\in\C$, $a\in\bS$  
                                \RETURN $\mathcal{A}\bigl((\C,\bS),d\bigr)$
			\end{algorithmic}
		\end{namedmech}

\begin{theorem}\label{thm:mech2:FneqC} \label{adsample-gen-thm}
Mechanism~\ref{alg:mech2L:FC} has $O\bigl(k\log \ell\log(1/\dt)\bigr)$ per-agent query
complexity, and achieves $O(1)$-distortion for the $\ell$-centrum problem with probability
at least $1-\dt$. 
\end{theorem}

\begin{proof}
There exists some $t_\ell\in\T$ such that 
$t^*_\ell\leq t_\ell\leq\max\{(1+\ve)t^*_\ell,\ve\cdot\frac{\OPT}{\ell}\}$.
By Theorem~\ref{adsample-aneqc-thm}, for this $t_\ell$, with probability at least $1-\dt$,
we obtain a $\bigl(76,35(1+\ve)\bigr)$ bicriteria solution. Hence, with probability at
least $1-\dt$, $\bS$ is such a bicriteria solution. We construct the entire metric on
the instance $\C\times\bS$, where $\C$ is the client-set and $\bS$ is the facility-set, so
we can run the $(5+\ve)$-approximation algorithm of~\cite{ChakrabartyS19} on this
instance. 
Let $T\sse\bS$ be the $\ell$-centrum solution returned. 
As in the proof of Theorem~\ref{adsample-thm}, letting $\OPT_{\C\times\bS}$ denote the
optimal $\ell$-centrum value for the $\C\times\bS$ instance, we have
$\OPT_{\C\times\bS}\leq 2\OPT+\topl(d(\C,\bS))$.  
So $\topl(d^*(\C,T))\leq\rho\cdot\OPT_{\C\times\bS}\leq 37\rho(1+\ve)\OPT$.
We argue that $T$ is a good
$\ell$-centrum solution for the original instance as well.

\medskip
\noindent \emph{Query Complexity}: The per-agent query
complexity of the modified $k$-center mechanism is $k$. 
We run Algorithm~\ref{adsample-aneqc-alg}, which also has $O(k)$ per-agent query
complexity, $O\bigl(|\T|\log(1/\dt)\bigr)$ times. So the number of queries per agent
incurred in this entire process is $O\bigl(k\log\ell\log(1/\dt)\bigr)$
Finally, we use $O(k)$ value queries per agent to compute the metric on $\C\times\bS$. 
Thus, the total number of queries per agent is 
$O\bigl(k\log\ell\log(1/\delta)\bigr)$.  
\end{proof}

\begin{proofof}{Theorem~\ref{adsample-aneqc-thm}}
We can borrow almost the entire proof of Theorem~\ref{adsample-topl-thm} (from
Section~\ref{adsample-topl-proof}), which proves 
the performance guarantee for adaptive sampling for $\ell$-centrum in the $\F=\C$ setting.
Let $S^*=\{c^*_1,\ldots,c^*_k\}\in\F^k$ be an optimal solution, and $C^*_1,\ldots,C^*_k$
denote the clusters induced by $S^*$.
The definitions of good, bad, close, far clusters, radius and core of a cluster remain
unchanged. The {\em only} portion of the proof of Theorem~\ref{adsample-topl-thm} that we
need to modify is the proof of Lemma~\ref{coreisgood} showing that hitting the core of a
cluster renders that cluster good. This also requires some changes to the parameters. We
take $\tau=38$, $\rho=35$, 
and $\betaval=\gbeta$, $\al=2$, $\gm=5$, $\kp=8$. These satisfy the following
inequalities:
\begin{equation}
\begin{split}
\betaval & \geq 3,\ \ \gm = 2\al+1\geq\betaval, \ \ \al>1, \ \
1-\tfrac{\gm}{\rho}\geq 2\cdot\tfrac{\kp+\betaval}{\rho} \\
\kp & \geq\al+\betaval+3, \quad
\Bigl(1-\tfrac{\gm}{\rho}\Bigr)\cdot\tfrac{\al-1}{2\al\kp}\geq\frac{1}{\tau}.
\end{split} 
\label{genchoices}
\end{equation}
The above inequalities are stronger than \eqref{abkchoices}, so almost the entire analysis
from the proof of Theorem~\ref{adsample-topl-thm}---in particular, Claim~\ref{allgood}, 
Lemmas~\ref{choosecore}--\ref{coreclose}---applies here as well. We only need to show the
following. 

\begin{claim} \label{gen-coreisgood}
Consider a cluster $C^*_q$ and let $S$ be the current center-set. 
Suppose that for some agent $s\in\core(C^*_q)$, we have that $\topalt(s)\in S$.  
Then $C^*_q$ is $\ell$-good (and hence remains $\ell$-good throughout). 
\end{claim}

\begin{proof}
Let $a=\topalt(s)$.
The quantity $\sum_{j \in C^*_q}(d(j,a) - \betaval t_\ell)^+$ is at most 
\begin{equation*}
\begin{split}
 \sum_{j\in C^*_q}\bigl(&(d(j,c^*_q)-t_\ell)^++(d(s,c^*_q)+d(s,a)-(\betaval-1)t_\ell)^+\bigr) \\
& \leq |C^*_q|\cdot\bigl(r_\ell(C^*_q)+(2d(s, c^*_q) - (\betaval-1)t_\ell)\bigr)^+. \\
& \leq |C^*_q|\bigl(r_\ell(C^*_q)+2\al\cdot r_\ell(C^*_q)\bigr).
\end{split}
\end{equation*}
The first inequality follows from the triangle inequality.
The second inequality follows from the definition of $r_\ell$, and since $d(s,a)\leq d(s,c^*_q)$.
The third is because $\betaval\geq 3$ and $s\in\core(C^*_q)$. Since $\gm\geq 2\al+1$, this
shows that $C^*_q$ is $\ell$-good.
\end{proof}

This completes the proof of Theorem~\ref{adsample-aneqc-thm}.
\end{proofof}


\section{Obtaining in-expectation guarantees}\label{sec:inExp}
The mechanisms presented so far achieve deterministic query-complexity upper bounds,
and distortion bounds that hold with high probability. We can easily modify our
mechanisms so that the distortion guarantees hold {\em in expectation},
without significantly increasing the query complexity. At a high level, the idea is to
simply set the failure probability to be sufficiently small, and in the case of failure,
return a solution that achieves bounded (but not necessarily $O(1)$) distortion, such as
the approximate $k$-center or $k$-median solution computed by Mechanism~\ref{kcenter}
or Mechanism~\ref{kmedian}.

We briefly discuss the changes to our mechanisms, focusing on the $\F=\C$ setting for
simplicity; the same ideas apply to the $\F\neq\C$ setting as well.

\begin{enumerate}[label=$\bullet$, leftmargin=*]
\item {\bf Modification of Mechanism~\ref{bb-meyerson}.}
We set $\delta = (\max\{k, \min\{\ell, \ln(k)n/\ell\}\})^{-1}$. 
If $\Sc=\es$ in step~\ref{line:mech1L-inexp2-start}, instead
of declaring failure, we let $\bS$ be the union of the solutions 
output by Mechanisms~\ref{kcenter} and~\ref{kmedian}, and continue.

The resulting mechanism achieves $O(1)$ expected distortion and has {\em expected}
per-agent query complexity $O(\log (\max\{k, \min\{\ell, \ln(k)n/\ell\}\})\log n)$. 
To see this, let $\error$ denote the ``bad event'' that $\Sc=\es$. 
The expected cost of the solution returned is at most 
\begin{equation*}
\OPT\cdot\Bigl[\bigl(1-\Pr[\error]\bigr)\cdot O(1)
+\Pr[\error]\cdot O\bigl(\min\{\ell,\ln (k)n/\ell\}\bigr)\Bigr]
\end{equation*}
since when $\error$ happens, $\bS$ 
is an $O\bigl(\min\{\ell,\ln (k)n/\ell\})\bigr)$-approximate solution, and this
approximation guarantee translates to the output (due to Lemma~\ref{sparsify}).

The expected per-agent query complexity bound follows because if $\error$ happens, then we
make at most $2k$ additional queries per-agent when running Mechanisms~\ref{kcenter}
and~\ref{kmedian}. 

\item {\bf Modification of Mechanism~\ref{adsample}.} 
We set $\dt=\bigl(\min\{\ell,\ln(k)n/\ell\}\bigr)^{-1}$, 
and initialize $\Sc$ in step~\ref{line:mech2l-inexp1-start} to include the outputs of
Mechanisms~\ref{kcenter} and~\ref{kmedian}.
This way, we are always guaranteed to return a solution of cost at most 
$O\bigl(\min\{\ell,\ln(k)n/\ell\}\bigr)\cdot\OPT$. 
So $O(1)$ distortion (i.e., cost $O(\OPT)$) with probability at least $1-\dt$, also
implies $O(1)$ expected distortion.  

The per-agent query complexity is deterministically bounded by
$\tO\bigl(k\log^2(\min\{\ell,n/\ell\})\bigr)$.

\item {\bf Modification of Mechanism~\ref{adsample-tot}.}
We simply set $\dt=1/\ell$. In Mechanism~\ref{adsample-tot}, every candidate
solution in $\Sc$ includes the output of Mechanism~\ref{kcenter}, and therefore has cost
at most $O(\ell)\cdot\OPT$. So $O(1)$ distortion with probability at least $1-\dt$, also
implies $O(1)$ expected distortion. The total query complexity is 
deterministically bounded by $O(k^2\log^2 n\log^2\ell)$. 
\end{enumerate}

it suffices to run the
mechanisms presented in Section \ref{sec:Blackbox} with a suitable success-probability
$\delta$; in the event that the mechanism \emph{fails}, we return an approximate
$k$-median or $k$-center solution instead. In particular, we will use Mechanisms
\ref{alg:AV} and \ref{alg:kcenter} to compute approximate $k$-median and $k$-center
solutions respectively.

\section{Conclusions} \label{concl}
We studied the $k$-committee election problem under the $\topl$ objective,
and devised constant-factor distortion mechanisms that achieve
$O(\log k\log n)$ and 
$\tO\bigl(k\log(\min\{\ell,n/\ell\})\bigr)$ per-agent query
complexity, 
and $O(k^2 \log^2 n\log\ell)$ total query complexity. 
Our logarithmic per-agent query-complexity bounds are obtained
via a versatile black-box reduction that reduces the ordinal problem to the 
cardinal setting using polylogarithmic number of per-agent queries.
The per-agent query-complexity bound independent of $n$ (for fixed $\ell$), and
the total-query complexity bound, are obtained via a novel sampling algorithm
that we develop for the $\ell$-centrum $k$-clustering problem.   

We consider value queries, but one could also consider other query models.
For instance, it may be easier for an agent to identify which candidates are at a distance
of at most $r$ from her location. We call such queries {\em ball queries}. 
Our black-box reduction (Mechanism~\bbtopl)
can in fact be implemented using $O(\log |\F|)$ ball
queries per agent, but computing an initial estimate of $\OPT$ becomes more difficult, 
as it is a non-trivial task to grasp the magnitude of the distances using relatively few ball
queries. One could also consider other types of queries (e.g. the threshold queries used
by \cite{maMenonLarson_thresholdQuery} or the comparison queries used by
\cite{amanatidisComparisonQuery}), or other sources of limited cardinal information.   
 

\FloatBarrier
\bibliographystyle{abbrv}
\bibliography{references}

\appendix 

 

\section{Proof of Theorem \ref{thm:AdaptiveSample-lcentrum-ring}} 
\label{app:aeqc}\label{app:AdaptiveSample}\label{app:MeyersonL}


The proof closely mirrors that of
Theorem~\ref{adsample-topl-thm}. We first observe that the ring-based implementation is
akin to using the earlier adaptive-sampling approach with {\em perturbed distances} $\td$
satisfying $d(j,S)\leq\td(j,S)\leq 2d(j,S)+\ve\OPT$ for every center-set $S$
encountered, and every $j\notin S$. We make this precise below.

\begin{lemma} \label{ring-perturbed}
Consider any iteration of Algorithm~\ref{adsample-ring}, and let $S$ be the set of centers
already chosen. For $j\in\C\sm S$, define $\td(j,S)=\thresh_h$ if
$j\in\ring_{\thresh_h}$. Then

\begin{enumerate}[label=(\alph*), topsep=0.2ex, itemsep=0.2ex, leftmargin=*]
\item $d(j,S)\leq\td(j,S)\leq 2d(j,S)+\ve\frac{\OPT}{n^2}$ for all $j\in\C\sm S$.

\item Algorithm~\ref{adsample-ring} chooses point $s_i$ in
line~\ref{adsample-ring-samppt} with probability 
$\frac{(\td(s_i,S)-\rbeta t_\ell)^+}{\sum_{j\in\C\sm S}(\td(j,S)-\rbeta t_\ell)^+}$.
\end{enumerate}
\end{lemma}

\begin{proof}
Part (a) is immediate the definition of the $\ring_{\thresh_h}$ rings, since the quantity
$B$ in line~\ref{kcentout} satisfies $B\leq 2\OPT$.

Fix some $w\in\C\sm S$, and 
let $\bthresh$ be such that $w\in\ring_{\bthresh}$. 
We have 
\begin{equation*}
\begin{split}
\Pr[s_i=w] & = \frac{|\ring_{\bthresh}|\cdot(\bthresh-\rbeta t_\ell)^+}
{\sum_{h=0}^N|\ring_{\thresh_h}|\cdot(\thresh_h-\rbeta t_\ell)^+}
\cdot\frac{1}{|\ring_{\bthresh}|} \\
& = \frac{(\td(w,S)-\rbeta t_\ell)^+}
{\sum_{h=0}^N\sum_{j\in\ring_{\thresh_h}}(\td(j,S)-\rbeta t_\ell)^+} 
= \frac{(\td(w,S)-\rbeta t_\ell)^+}{\sum_{j\in\C\sm S}(\td(j,S)-\rbeta t_\ell)^+}.
\qedhere
\end{split}
\end{equation*}
\end{proof}

Given Lemma~\ref{ring-perturbed}, we can essentially carry over all the arguments in the
proof of  Theorem~\ref{adsample-topl-thm} by working with the perturbed $\td$ distances.
But we do need to rework the arguments and make relatively minor changes to account for
the perturbation. This also necessitates changes 
to the values of the parameters $\al,\betaval,\gm,\kp$, and $\tau,\rho$ used in the
analysis. 

Lemma~\ref{ring-perturbed} also easily implies the second portion of the theorem statement
regarding the quality of the estimate. 
Note that the estimate is {\em precisely}
$\topl(\td(\C,S))$, where we define $\td(j,S)=0$ for $j\in S$. So by the relationship
between $\td$ and $d$, we have that the estimate is at least $\topl(d(\C,S))$ and at most 
$2\topl(d(\C,S))+\ell\ve\cdot\frac{\OPT}{n^2}$. 

\medskip
We set $\tau=62$, $\rho=50$, and take 
$\betaval=\rbeta$, $\alpha=2$, $\gm=3$, and $\kappa=9$; they are chosen to satisfy the
following inequalities: 
\begin{equation}
\begin{split}
\betaval & =2\cdot\rbetap\geq 3, \quad \gm=\al+1\geq\frac{\betaval}{2}, \quad \al>1 \\
1-\tfrac{2\gm}{\rho} & \geq 2\cdot\tfrac{2\kp+\betaval}{\rho}, \quad
\kp \geq\al+\betaval+3, \quad
\Bigl(1-\tfrac{2\gm}{\rho}\Bigr)\cdot\tfrac{\al-1}{3\al\kp} \geq\frac{1}{\tau}.
\end{split} 
\label{ringchoices}
\end{equation}

Let $S$ be the current
center-set. Consider a cluster $C^*_q$ with center $c^*_q$. 
We now define:
\begin{enumerate}[label=$\bullet$, topsep=0.4ex, itemsep=0.1ex, leftmargin=*]
\item $C^*_q$ is $\ell$-good if 
$\sum_{j \in C^*_q}(d(j,S)-\rbetap t_\ell)^+ \leq 
\gamma\bigl[\sum_{j \in C^*_q}(d(j, c^*_q) - t_\ell)^++\frac{\ve\OPT}{n}\bigr]$,
otherwise it is $\ell$-bad;
\item $r_\ell(C^*_q) = \frac{\sum_{j \in C^*_q} (d(j, c^*_q)-t_\ell)^++\ve\OPT/n}{|C^*_q|}$;
$C^*_q$ is $\ell$-close if $d(c^*_q,S)\leq\kp\max\{t_\ell,r_\ell(C^*_q)\}$, and is
$\ell$-far otherwise;
\item $\core_\ell(C^*_q)$ is
\begin{equation*}
\begin{cases}
\{j \in C^*_q : d(j, c^*_q) \leq t_\ell\}; & \text{if $C^*_q$ is $\ell$-close} \\ 
\{j \in C^*_q : (d(j, c^*_q) - t_\ell)^+ \leq \alpha \cdot r_\ell(C^*_q)\};
& \text{otherwise}.
\end{cases}
\end{equation*}
\end{enumerate}

Similar to Claims~\ref{allgood} and~\ref{coreisgood}, we have the following. 
\begin{lemma}
The following hold.
\begin{enumerate}[label=(\alph*), topsep=0.1ex, itemsep=0.1ex, leftmargin=*]
\item If every cluster is $\ell$-good, then 
$\topl(d(\mathcal{C}, S)) \leq(1+2\varepsilon)\gamma\cdot\topl(d(\mathcal{C}, S^*)).$ 
\item Let $S$ be the current center-set. If $S\cap\core(C^*_q)\neq\es$ for some cluster
  $C^*_q$, then $C^*_q$ is $\ell$-good (and hence remains $\ell$-good throughout). 
\end{enumerate}
\end{lemma}

\begin{proof}
Part (a) follows from exactly the same arguments as in the proof of
Claim~\ref{allgood}. For part (b), as in the proof of Claim~\ref{coreisgood}, if 
$s\in S\cap\core(C^*_q)$, then we have 
$\sum_{j \in C^*_q}(d(j,s) - \rbetap t_\ell)^+\leq 
\sum_{j \in C^*_q} (d(j, c^*_q) - t_\ell)^++\al|C^*_q|r_\ell(C^*_q)$. 
Plugging in $r_\ell(C^*_q)$, 
this again shows that $C^*_q$ is $\ell$-good.
\end{proof}

We prove analogues of
Lemmas~\ref{badprob}--\ref{coreclose}, and Lemma~\ref{choosecore}, which involves
reworking the arguments with the $\td$ distances.
Consider an iteration $i$, and let $S=S_{i-1}$ denote the current center-set. 
Suppose we have $\topl(d(\C,S))>\rho(1+2\ve)\OPT$.
Recall that $t^*_\ell\leq t_\ell\leq\max\bigl\{(1+\ve)t^*_\ell,\ve\cdot\frac{\OPT}{\ell}\bigr\}$,
and $d(j,S)\leq \td(j,S)\leq 2d(j,S)+\ve\cdot\frac{\OPT}{n^2}$ for all $j\in\C\sm S$.
Let the sampled point $s^*_i$ belong to cluster $Z^*\in\{C^*_1,\ldots,C^*_q\}$.

As before, $\good$, $\bad$, $\close$, $\far\sse[k]$ denote the index-sets of \{good, bad,
close, far\} clusters respectively.

\begin{lemma} \label{ring-badprob}
$\Pr[Z^*\text{ is $\ell$-bad}] \geq  1-\frac{2\gamma}{\rho}$.
\end{lemma}

\begin{proof} 
$\Pr[Z^*\text{ is $\ell$-good}]$ is
$\frac{\sum_{q\in\good}\sum_{j\in C^*_q}(\td(j,S)-\betaval t_\ell)^+}{\sum_{j\in\C}(\td(j,S)-\betaval t_\ell)^+}$,
which is at most
\[
\frac{\betaval t_{\ell } \cdot \ell+\sum_{q\in\good}\sum_{j\in C^*_q}(2d(j, c^*_q)-\betaval
  t_\ell)^++\ve\OPT}
{\betaval t_{\ell} \cdot \ell + \sum_{j \in \mathcal{C}} (d(j, S) - \betaval t_{\ell})^+}.
\]
The denominator above is at least $\topl(d(\C,S))$, by Claim~\ref{csproxy} (b), and so at
least $\rho(1+2\ve)\OPT$. We upper bound the numerator.
By the definition of $\ell$-good clusters and since $\betaval=2\cdot\rbetap$, for any
$q\in\good$, we have 
$\sum_{j\in C^*_q}(2d(j, c^*_q)-\betaval t_\ell)^+\leq 2\gm\sum_{j\in C^*_q}(d(j,c^*_q)-t_\ell)^+$.
So, since $t_\ell\geq t^*_\ell$, the above expression is at most 
\begin{gather*}
\frac{\betaval\max\{(1+\ve)\ell t^*_\ell,\ve\OPT\}+2\gm\sum_{j\in\C}(d(j,S^*)-t^*_\ell)^++\ve\OPT}
{\rho(1+2\ve)\OPT}
\end{gather*}
which is at most $\frac{2\gm(1+2\ve)}{\rho(1+2\ve)}$, 
where we use that $2\gm\geq\betaval$.
\end{proof} 

\begin{lemma} \label{ring-corebadfar}
Consider any $\ell$-far cluster $C^*_q$.
$\Pr[s_i\in\core(Z^*)\,|\,Z^*=C^*_q]\geq\frac{\al-1}{3\al\kp}$.
\end{lemma}

\begin{proof} 
The probability is $\Pr[s_i\in\core_\ell(C^*_q)]/\Pr[s_i\in C^*_q]$.
We abbreviate $r_\ell(C^*_q)$ to $r_\ell$ in this proof.
we have $|\core_\ell(C^*_q)| \geq \frac{\alpha - 1}{\alpha} \cdot |C^*_q|$. 
So
\begin{alignat}{1}
\frac{\Pr[s_i\in\core_{\ell}(C^*_q)]}{\Pr[s_i\in C^*_q]}
& = \frac{\sum_{j \in\core_{\ell}(C^*_q)}(\td(j,S) - \betaval t_{\ell})^+}
{\sum_{j \in C^*_q} (\td(j,S) - \betaval t_\ell)^+} 
\geq\frac{\sum_{j \in\core_\ell(C^*_q)}(d(c^*_q,S) - d(j,c^*_q) - \betaval t_\ell)^+}
{\sum_{j \in C^*_q} (2d(j,c^*_q) + 2d(c^*_q, S) - \betaval t_\ell)^++\frac{\ve\OPT}{n}} 
\notag \\
& \geq\frac{|\core_\ell(C^*_q)| \cdot (d(c^*_q, S) -\alpha  r_{\ell} -(\betaval+1)t_{\ell})}
{2|C^*_q|\cdot r_{\ell}+2|C^*_q|\cdot\bigl(d(c^*_q,S)-t_{\ell}\bigr)^++\frac{\ve\OPT}{n}}
\notag \\
& \geq\frac{|\core_\ell(C^*_q)| \cdot (d(c^*_q, S) -\alpha  r_{\ell} -(\betaval+1)t_{\ell})}
{3|C^*_q|\bigl(r_{\ell}+d(c^*_q,S)-t_{\ell}\bigr)}
\notag \\
& \geq \frac{\alpha - 1}{3\alpha} \cdot 
\frac{d(c^*_q, S) -\alpha r_{\ell} - (\betaval+1)t_{\ell}}{r_{\ell} + d(c^*_q,S)}
\label{coreineq1}
\end{alignat}
%
The second inequality is because $d(j, c^*_q) - t_\ell \leq\al r_\ell$ for all 
$j\in\core_\ell(C^*_q)$ and $\betaval=\rbeta$; 
the third inequality is because $\frac{\ve\OPT}{n}\leq|C^*_q|r_\ell$ and  
$d(c^*_q,S)\geq\kp\max\{t_\ell,r_\ell\}\geq t_\ell$,
as $\kp\geq\al+\betaval+1$. 
Expression \eqref{coreineq1} is an increasing function of $d(c^*_q,S)$, and so since
$C^*_q$ is $\ell$-far, we can lower bound 
$\frac{d(c^*_q,S) -\alpha  r_{\ell} - (\betaval+1)t_{\ell}}{r_{\ell} + d(c^*_q,S)}$
by $\frac{1}{\kp}$ exactly as in the proof of Lemma~\ref{corebadfar}.
%
\end{proof} 	

\begin{lemma} \label{ring-coreclose}
$\Pr[Z^*\text{ is $\ell$-close}, s_i\notin\core(Z^*)]\leq\frac{\kp+\betaval}{\rho}$.
\end{lemma}

\begin{proof} 
The given probability is 
\begin{equation*}
\frac{\sum_{q\in\close} \sum_{j \in C^*_q\sm\core_{\ell}(C^*_q)} 
(\td(j, S) - \betaval t_{\ell}) ^+ }{\sum_{j \in \mathcal{C}} (\td(j, S) - \betaval t_\ell)^+ }
\leq \frac{\ell\cdot\betaval t_\ell+\sum_{q\in\close} \sum_{j \in C^*_q\sm\core_{\ell}(C^*_q)} (2d(j, S) - \betaval
  t_{\ell}) ^+ +\ve\OPT}{\ell\cdot\betaval t_\ell+\sum_{j \in \mathcal{C}} (d(j, S) -
  \betaval t_\ell)^+ }
\end{equation*}

\vspace*{-2ex}
\begin{alignat*}{1}
\quad 
& \leq \frac{\ell\cdot\betaval t_\ell+2\sum_{q\in\close}\sum_{j \in
    C^*_q\sm\core_{\ell}(C^*_q)} 
\bigl(d(j,c^*_q)+d(c^*_q,S)-\rbetap t_{\ell}\bigr)^+ +\ve\OPT}{\rho(1+2\ve)\OPT} \\
& \leq \frac{\ell\cdot\betaval
  t_\ell+2\sum_{j\in\C}\bigl(d(j,S^*)-t_{\ell}\bigr)^++\ve\OPT}{\rho(1+2\ve)\OPT}
+\frac{\sum_{q\in\close}|C^*_q\sm\core(C^*_q)|(d(c^*_q,S)-t_\ell)^+}{\rho(1+2\ve)\OPT}.
\end{alignat*}
The second term above is at most $\frac{\kp}{\rho}$ due to the same reasoning as in the
proof of Lemma~\ref{coreclose}. 
The first term is at most $\frac{\betaval}{\rho}$ since
\begin{equation*}
\begin{split}
& \ell\cdot\betaval t_\ell+2\sum_{j\in\C}\bigl(d(j,S^*)-t_{\ell}\bigr)^++\ve\OPT \\
& \leq\betaval\max\{\ell t^*_\ell,\ve\OPT\}+2\sum_{j\in\C}\bigl(d(j,S^*)-t^*_{\ell}\bigr)^++\ve\OPT.
\end{split}
\end{equation*}
If $\ell t^*_\ell\geq\ve\OPT$, then the last expression is at most
$\betaval\bigl[\ell
t^*_\ell+\sum_{j\in\C}\bigl(d(j,S^*)-t^*_\ell\bigr)^+\bigr]+\ve\OPT\leq\betaval(1+\ve)\OPT$.
Otherwise, this expression is at most 
$\betaval\ve\OPT+2\sum_{j\in\C}\bigl(d(j,S^*)-t^*_\ell\bigr)^++\ve\OPT\leq\betaval(1+2\ve)\OPT$.
Putting the two bounds together, we obtain that
$\Pr[Z^*\text{ is $\ell$-close}, s_i\notin\core(Z^*)]\leq\frac{\kp+\betaval}{\rho}$.
\end{proof}

Finally, we combine the bounds given by Lemma~\ref{ring-badprob}--\ref{ring-coreclose} in
the same way as before to obtain that 
$\Pr[Z^*\text{ is $\ell$-bad},\,s_i\in\core_\ell(Z^*)]$ is at least
$\bigl(1-\frac{2\gm}{\rho}\bigr)\cdot\frac{\al-1}{3\al\kp}\geq\frac{1}{\tau}$.

Then, by the same martingale argument used in the proof of
Theorem~\ref{adsample-topl-thm}, we obtain that the center-set computed after   
$\ceil{\tau(k+\sqrt{k})}\leq 124k$ iterations satisfies the stated
approximation guarantee with probability at least $1-e^{-\frac{1}{4\tau}}$.
\hfill \qed

\end{document}